\documentclass[twocolumn,10pt]{IEEEtran}

\usepackage{amssymb}
\usepackage{amsmath}
\usepackage{amsthm}
\usepackage{graphicx}
\usepackage{comment}
\usepackage{xcolor}
\usepackage{subfigure}
\usepackage{hyperref}

\usepackage{algorithm}
\usepackage{algorithmic}

\usepackage[acronym]{glossaries}
\newacronym{siso}{SISO}{single-input single-output}
\newacronym{star-ris}{STAR-RIS}{simultaneously transmitting and reflecting RIS}
\newacronym{iid}{i.i.d.}{independent and identically distributed}

\newacronym{los}{LoS}{line-of-sight}
\newacronym{mimo}{MIMO}{multiple-input multiple-output}

\newacronym{em}{EM}{electromagnetic}

\newacronym{etsi}{ETSI}{European telecommunications standards institute}

\newtheorem{proposition}{Proposition}

\begin{document}

\title{Physics-Compliant Modeling and Scaling Laws of\\Multi-RIS Aided MIMO Systems}

\author{Matteo~Nerini,~\IEEEmembership{Member,~IEEE},
        Gabriele~Gradoni,~\IEEEmembership{Member,~IEEE},
        Bruno~Clerckx,~\IEEEmembership{Fellow,~IEEE}

\thanks{Part of this work has been presented in the 19th European Conference on Antennas and Propagation, 2025 \cite{ner24}.}
\thanks{This work was supported in part by UKRI under Grant EP/Y004086/1, EP/X040569/1, EP/Y037197/1, EP/X04047X/1, and EP/Y037243/1.}
\thanks{Matteo Nerini and Bruno Clerckx are with the Department of Electrical and Electronic Engineering, Imperial College London, SW7 2AZ London, U.K. (e-mail: m.nerini20@imperial.ac.uk; b.clerckx@imperial.ac.uk).}
\thanks{Gabriele Gradoni is with the Institute for Communication Systems, University of Surrey, GU2 7XH Guildford, U.K. (e-mail: g.gradoni@surrey.ac.uk).}}

\maketitle

\begin{abstract}
Reconfigurable intelligent surface (RIS) enables the control of wireless channels to improve coverage.
To further extend coverage, multi-RIS aided systems have been explored, where multiple RISs steer the signal via a multi-hop path.
However, deriving a physics-compliant channel model for multi-RIS aided systems is still an open problem.
In this study, we fill this gap by modeling multi-RIS aided systems through multiport network theory, and deriving a channel model accounting for impedance mismatch, mutual coupling, and structural scattering.
The derived physics-compliant model differs from the model widely used in literature, which omits the RIS structural scattering.
To quantify this difference, we derive the channel gain scaling laws of the two models under line-of-sight (LoS) and multipath channels.
Theoretical insights, validated by numerical results, show an important discrepancy between the physics-compliant and the widely used models, increasing with the number of RISs and multipath richness.
In a multi-hop system aided by four 128-element RISs with multipath channels, optimizing the RISs using the widely used model and applying their solutions to the physics-compliant model achieves only 7\% of the maximum channel gain.
This highlights how severely mismatched channel models can be, calling for more accurate models in communication theory.
\end{abstract}

\glsresetall


\section{Introduction}
\label{sec:intro}

Reconfigurable intelligent surface (RIS) has emerged as a technology enabling dynamic control over the \gls{em} propagation environment in wireless networks \cite{wu21}.
RIS technology leverages surfaces made of elements with programmable scattering properties to manipulate impinging \gls{em} signals, thereby enhancing the channel strength and extending coverage.
While most of the literature on RIS focuses on systems aided by a single RIS, multi-RIS aided systems, also known as multi-hop RIS-aided systems, have attracted attention as they can further enhance coverage and overcome multiple obstacles \cite[Section~8.3]{etsi25}.
Initial studies considered the optimization and power scaling analysis of double-RIS aided systems \cite{han20,zhe21-1}, where the signal reaches the receiver following a double-hop path.
The capacity maximization problem for double-RIS aided systems has been tackled in \cite{han22,an23}.
In addition, a channel estimation protocol for double-RIS systems has been presented in \cite{zhe21-2}, while a geometry-based stochastic channel model has been proposed in \cite{qi23}.
As a generalization of double-RIS aided systems, multi-RIS aided systems have been studied, where multiple cooperative RISs are deployed to drive the \gls{em} signal towards the intended location through a multi-hop path \cite{mei21,hua21}.
In these systems, multiple RISs can boost the channel strength in harsh propagation conditions, with consequent coverage extension \cite{mei22-2}.
Other benefits of multi-RIS aided systems include the ability to artificially create multipath beams in single-user systems \cite{mei22-1}, and to maximize the sum rate in multi-user systems by mitigating interference \cite{ma22,ngu23}.

Works \cite{han20}-\cite{ngu23} deployed multiple reflective RISs to extend the coverage, each of them covering half-space and characterized by a diagonal phase shift matrix.
Given its mathematical constraint, we refer to this conventional RIS architecture as diagonal RIS (D-RIS).
Besides, to further extend the coverage of RIS-aided systems, more flexible RIS architectures have been proposed under the umbrella term of beyond diagonal RIS (BD-RIS), which are characterized by a matrix allowed to have non-zero off-diagonal entries \cite{li23}.
BD-RIS offers two ways to improve the coverage.
First, reflective BD-RISs can improve coverage and performance through their advanced flexibility, allowing for more versatile manipulation of the signal.
Efficient reflective BD-RIS architectures have been proposed, such as fully- and group-connected RISs \cite{she20,ner22}, and tree- and forest-connected RISs \cite{ner23-1}.
Second, BD-RIS enables the signal impinging in one element to be irradiated by other elements, allowing the transmission of the signal through the RIS.
Thus, BD-RIS architectures with transmissive capabilities have been proposed to reach a 360$^{\circ}$ coverage, such as by using \gls{star-ris} \cite{xu21}, and the more general hybrid BD-RIS and multi-sector BD-RIS architectures \cite{li22-1,li22-2}.
In \cite{jav24}, multiple transmissive RISs have been used to serve mobile users through a multi-hop path.

Accurate modeling of RIS-aided wireless channels is crucial for designing and optimizing RIS-aided systems, including a single or multiple RISs implemented through D-RIS or BD-RIS architectures \cite{per23,del25}.
To rigorously model wireless channels in the presence of a single RIS, multiport network theory has been successfully utilized \cite{gra21,wil22,she20}.
Specifically, previous works used multiport network theory to derive physics-compliant RIS-aided channel models accounting for the impedance mismatch and mutual coupling effects at the transmitter, receiver, and RIS.
Different models have been proposed based on three equivalent formalisms, i.e., impedance (or $Z$) parameters \cite{gra21}, admittance (or $Y$) parameters \cite{wil22}, and scattering (or $S$) parameters \cite{she20}.
The relationship between impedance and scattering parameters has been more recently analyzed in \cite{nos24-1,nos24-2,li24,abr23}, and a universal framework has been derived in \cite{ner23} highlighting the connection between impedance, admittance, and scattering parameters.

While substantial effort has been devoted to system-level optimization of multi-RIS aided systems, the rigorous modeling of multi-RIS aided channels is still an open issue.
Channel models conventionally used in previous works on multi-RIS have not been derived by first physics principles, hence they are hard to validate and their limit of validity is unclear.
To fill this gap, in this study, we model the channel of multi-RIS aided \gls{mimo} systems through multiport network theory.
We derive the expression of a physics-compliant channel model and the scaling law of the achievable channel gain, under \gls{los} and multipath channels.

\textit{Contributions}:
The contributions of this study are as follows.

\textit{First}, we derive a physics-compliant channel model for multi-RIS aided \gls{mimo} systems by using multiport network theory, clarifying its underlying assumptions.
The derived channel model accounts for the impedance mismatch and mutual coupling at the transmitter, receiver, and RISs, the effects of all the channels between the transmitter/receiver and the RISs, and the structural scattering of the RISs\footnote{In antenna theory, the structural scattering is a component of the field scattered by an antenna \cite[Chapter 2]{bal16}. In the context of RIS, the structural scattering of a RIS results in a specular reflection at the RIS independent from the RIS reconfiguration \cite{abr23}.}.
As these effects are commonly neglected in related literature, our channel model is crucial to better understand multi-hop communication links enabled by multi-RIS aided systems.

\textit{Second}, we simplify the derived channel model by assuming perfect matching and no mutual coupling at the transmitter, receiver, and RISs for the sake of tractability.
The obtained simplified channel model is aligned with the model widely used in literature, while it introduces a slight variation to incorporate the structural scattering of the RISs, which is typically ignored in the literature.

\textit{Third}, we analyze the effect of hybrid and multi-sector BD-RISs working in transmissive mode on the multi-RIS aided channel expression.
Interestingly, we show that the structural scattering of RISs working in transmissive mode does not impact the channel model.
Thus, the widely used channel model is physics-compliant in the case of multi-RIS aided systems where all the RISs are used in transmissive mode.

\textit{Fourth}, we compare the derived physics-compliant channel model with the widely used one, when all the RISs are used in reflective mode and assuming \gls{los} channels.
To this end, we provide the scaling laws of the achievable channel gains for the two models.
Theoretical insights, supported by numerical simulations, show that the relative difference between the two channel gains grows with the number of RISs and decreases with the number of RIS elements.
In a system aided by four RISs with 128 elements each, this difference is higher than 80\%, and by optimizing the RISs based on the widely used model it can be achieved only 56\% of the maximum channel gain of the physics-compliant model.

\textit{Fifth}, we analyze the physics-compliant and widely used channel models under multipath channels.
Specifically, we propose an algorithm to maximize the channel gain by optimizing D-RIS and BD-RIS architectures, and derive closed-form upper bounds on the channel gains.
Numerical results, supported by theoretical intuition, show that the discrepancy between the physics-compliant and widely used channel models is more severe under multipath channels compared to \gls{los} channels.
In a system aided by four RISs with 128 elements each, the relative difference between the two channel gains is higher than 1000\% with Rayleigh channels.
Consequently, by optimizing the RISs based on the widely used model, only 7\% of the maximum channel gain of the physics-compliant model can be achieved.
An important takeaway message from this work to the communication society is to encourage the integration of more accurate, physics-based, channel models derived from first \gls{em} principles into communication theoretic analysis.

\textit{Organization}:
In Section~\ref{sec:multiport}, we characterize multi-RIS aided systems with multiport network theory.
In Section~\ref{sec:model}, we derive a physics-compliant model of multi-RIS aided systems accounting for impedance mismatch and mutual coupling.
In Section~\ref{sec:model-no-MC}, we simplify the model and show that it differs from the widely used channel model.
In Section~\ref{sec:model-multi-sector}, we model multi-RIS aided systems in the presence of both reflective and transmissive RISs.
In Sections~\ref{sec:gain-los} and \ref{sec:gain-multipath}, we maximize the channel gain of the physics-compliant and widely used models, under \gls{los} and multipath channels, respectively.
In Section~\ref{sec:results}, we provide numerical results to validate the theoretical insights.
In Section~\ref{sec:conclusion}, we conclude this work.
The simulation code is available at \href{https://github.com/matteonerini/multi-ris}{https://github.com/matteonerini/multi-ris}.

\textit{Notation}:
Vectors and matrices are denoted with bold lower and bold upper letters, respectively.
Scalars are represented with letters not in bold font.
$\Re\{a\}$, $\Im\{a\}$, $\vert a\vert$, and arg$(a)$ refer to the real part, imaginary part, absolute value, and phase of a complex scalar $a$, respectively.
$\mathbf{a}^T$, $[\mathbf{a}]_{i}$, and $\Vert\mathbf{a}\Vert$ refer to the transpose, $i$th element, and $l_{2}$-norm of a vector $\mathbf{a}$, respectively.
$\mathbf{A}^T$, $[\mathbf{A}]_{i,j}$, $\Vert\mathbf{A}\Vert$, and $\Vert\mathbf{A}\Vert_F$ refer to the transpose, $(i,j)$th element, spectral norm, and Frobenius norm of a matrix $\mathbf{A}$, respectively.
$\mathbb{R}$ and $\mathbb{C}$ denote real and complex number sets, respectively.
$j=\sqrt{-1}$ denotes the imaginary unit.
$\mathbf{0}$ and $\mathbf{I}$ denote an all-zero matrix and an identity matrix with appropriate dimensions, respectively.
$\mathcal{CN}(\mathbf{0},\mathbf{I})$ denotes the distribution of a circularly symmetric complex Gaussian random vector with mean vector $\mathbf{0}$ and covariance matrix $\mathbf{I}$ and $\sim$ stands for ``distributed as''.
diag$(a_1,\ldots,a_N)$ refers to a diagonal matrix with diagonal elements being $a_1,\ldots,a_N$.
diag$(\mathbf{A}_{1},\ldots,\mathbf{A}_{N})$ refers to a block diagonal matrix with blocks being $\mathbf{A}_{1},\ldots,\mathbf{A}_{N}$.

\section{Multiport Network Theory}
\label{sec:multiport}

Consider a \gls{mimo} communication system between an $N_T$-antenna transmitter and an $N_R$-antenna receiver aided by $L$ RISs, each having $N_I$ elements, as represented in Fig.~\ref{fig:diagram}.
Following previous literature \cite{she20,gra21,wil22,nos24-1,nos24-2,li24,abr23,ner23}, we model the wireless channel as an $N$-port network, with $N=N_T+LN_I+N_R$.

\begin{figure}[t]
\centering
\includegraphics[width=0.48\textwidth]{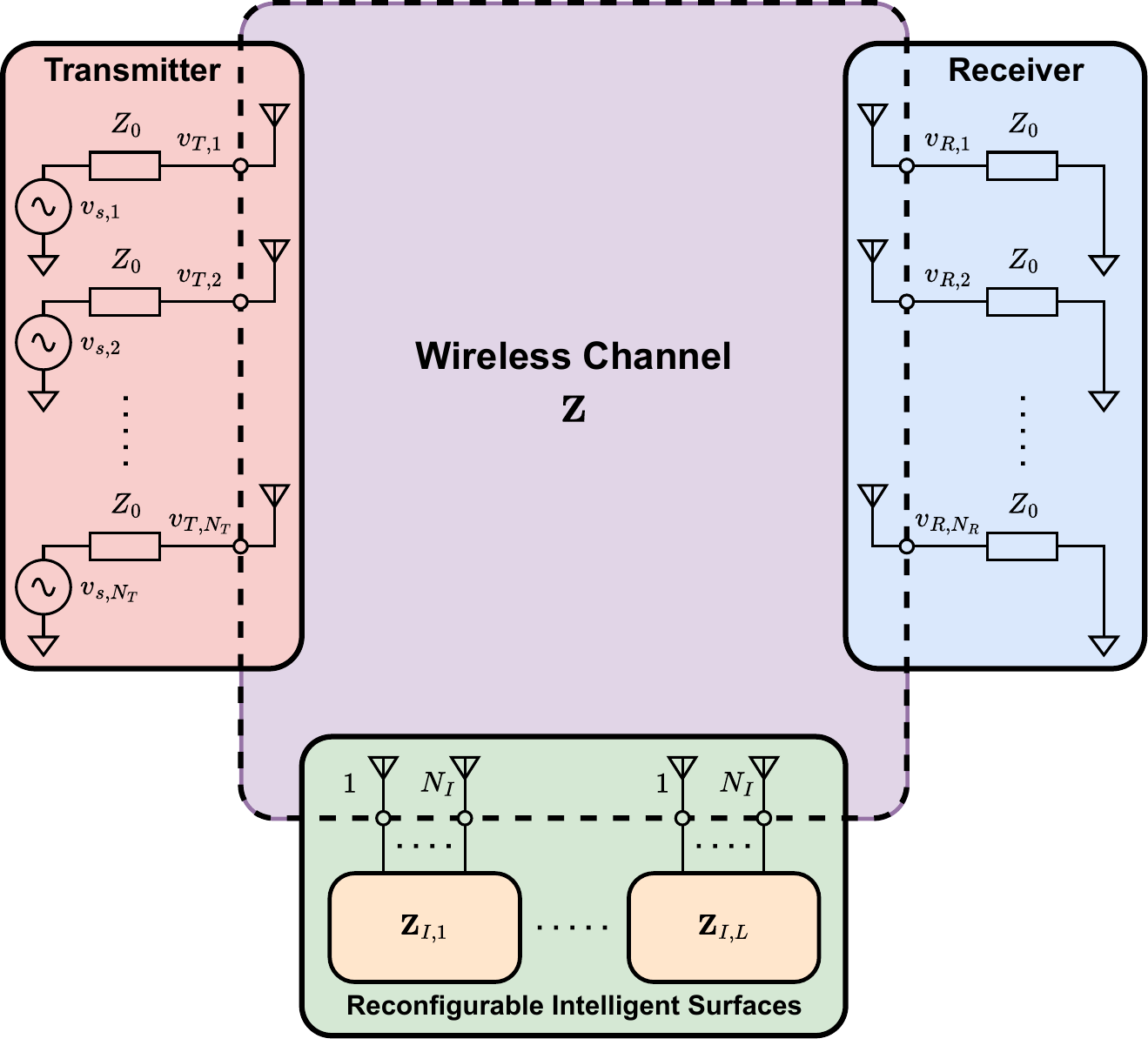}
\caption{Multi-RIS aided MIMO system modeled through multiport network theory.}
\label{fig:diagram}
\end{figure}

According to multiport network theory \cite[Chapter 4]{poz11}, the $N$-port network modeling the wireless channel can be characterized by its impedance matrix $\mathbf{Z}\in\mathbb{C}^{N\times N}$, which can be partitioned as
\begin{equation}
\mathbf{Z}=
\begin{bmatrix}
\mathbf{Z}_{TT} & \mathbf{Z}_{TI} & \mathbf{Z}_{TR}\\
\mathbf{Z}_{IT} & \mathbf{Z}_{II} & \mathbf{Z}_{IR}\\
\mathbf{Z}_{RT} & \mathbf{Z}_{RI} & \mathbf{Z}_{RR}
\end{bmatrix}.\label{eq:Z}
\end{equation}
In this partition, we identify the antennas at the transmitter, RISs, and receiver with subscripts $T$, $I$, and $R$, respectively.
Accordingly, $\mathbf{Z}_{XX}$ is the impedance matrix of the antenna array at $X$, where $X\in\{T,I,R\}$, and $\mathbf{Z}_{XY}$ is the transmission impedance matrix from $Y$ to $X$, where $X,Y\in\{T,I,R\}$.
More in detail, $\mathbf{Z}_{TT}\in\mathbb{C}^{N_T\times N_T}$ and $\mathbf{Z}_{RR}\in\mathbb{C}^{N_R\times N_R}$ refer to the impedance matrices of the antenna arrays at the transmitter and receiver, respectively, whose diagonal entries refer to the antenna self-impedance while the off-diagonal entries refer to antenna mutual coupling.
The matrix $\mathbf{Z}_{II}\in\mathbb{C}^{LN_I\times LN_I}$ can be partitioned as
\begin{equation}
\mathbf{Z}_{II}=
\begin{bmatrix}
\mathbf{Z}_{II,1} & \mathbf{Z}_{1,2} & \cdots & \mathbf{Z}_{1,L}\\
\mathbf{Z}_{2,1} & \mathbf{Z}_{II,2} & \cdots & \mathbf{Z}_{2,L}\\
\vdots & \vdots & \ddots & \vdots\\
\mathbf{Z}_{L,1} & \mathbf{Z}_{L,2} & \cdots & \mathbf{Z}_{II,L}
\end{bmatrix},
\end{equation}
where $\mathbf{Z}_{II,\ell}\in\mathbb{C}^{N_I\times N_I}$ refer to the impedance matrix of the antenna array at the $\ell$th RIS, whose diagonal entries refer to the antenna self-impedance while the off-diagonal entries refer to antenna mutual coupling, and $\mathbf{Z}_{i,j}\in\mathbb{C}^{N_I\times N_I}$ refer to the transmission impedance matrix from the $j$th RIS to the $i$th RIS.
Accordingly, $\mathbf{Z}_{IT}\in\mathbb{C}^{LN_I\times N_T}$ can be partitioned as
\begin{equation}
\mathbf{Z}_{IT}=\left[\mathbf{Z}_{IT,1}^T,\mathbf{Z}_{IT,2}^T,\ldots,\mathbf{Z}_{IT,L}^T\right]^T,
\end{equation}
where $\mathbf{Z}_{IT,\ell}\in\mathbb{C}^{N_I\times N_T}$ is the transmission impedance matrix from the transmitter to the $\ell$th RIS.
$\mathbf{Z}_{RI}\in\mathbb{C}^{N_R\times LN_I}$ can be partitioned as
\begin{equation}
\mathbf{Z}_{RI}=\left[\mathbf{Z}_{RI,1},\mathbf{Z}_{RI,2},\ldots,\mathbf{Z}_{RI,L}\right],
\end{equation}
where $\mathbf{Z}_{RI,\ell}\in\mathbb{C}^{N_R\times N_I}$ is the transmission impedance matrix from the $\ell$th RIS to the receiver.
$\mathbf{Z}_{RT}\in\mathbb{C}^{N_R\times N_T}$ refer to the transmission impedance matrix from the transmitter to the receiver.
Similarly, $\mathbf{Z}_{TI}\in\mathbb{C}^{N_T\times LN_I}$, $\mathbf{Z}_{IR}\in\mathbb{C}^{LN_I\times N_R}$, and $\mathbf{Z}_{TR}\in\mathbb{C}^{N_T\times N_R}$ refer to the transmission impedance matrices from the RISs to transmitter, and from the receiver to RISs, and from the receiver to transmitter, respectively.
For reciprocal channels, it holds $\mathbf{Z}_{TR}=\mathbf{Z}_{RT}^T$, $\mathbf{Z}_{TI}=\mathbf{Z}_{IT}^T$, and $\mathbf{Z}_{IR}=\mathbf{Z}_{RI}^T$.

At the transmitter, the $n_T$th transmitting antenna is connected in series with a source voltage $v_{s,n_T}\in\mathbb{C}$ and a source impedance $Z_0$, e.g., set to $Z_0=50\:\Omega$, and we denote the voltage at the $n_T$th transmitting antenna as $v_{T,n_T}\in\mathbb{C}$, for $n_T=1,\ldots,N_T$.
At the $\ell$th RIS, the $N_I$ antennas are connected to an $N_I$-port reconfigurable impedance network with impedance matrix denoted as $\mathbf{Z}_{I,\ell}\in\mathbb{C}^{N_I\times N_I}$, for $\ell=1,\ldots,L$.
At the receiver, all the receiving antennas are connected in series with a load impedance $Z_0$, and we denote the voltage at the $n_R$th receiving antenna as $v_{R,n_R}\in\mathbb{C}$, for $n_R=1,\ldots,N_R$, as shown in Fig.~\ref{fig:diagram}.

In this study, our goal is to characterize the expression of the channel matrix $\mathbf{H}\in\mathbb{C}^{N_R\times N_T}$ relating the transmitted signal $\mathbf{v}_T=[v_{T,1},\ldots,v_{T,N_T}]^T\in\mathbb{C}^{N_T\times 1}$ and the received signal $\mathbf{v}_R=[v_{R,1},\ldots,v_{R,N_R}]^T\in\mathbb{C}^{N_R\times 1}$ through
\begin{equation}
\mathbf{v}_R=\mathbf{H}\mathbf{v}_T,
\end{equation}
when the communication is aided by the $L$ RISs.

\section{Multi-RIS Aided Channel Model}
\label{sec:model}

As highlighted in previous work, the channel model for $\mathbf{H}$ derived through multiport network theory with no assumptions is difficult to interpret \cite{she20,gra21,wil22,ner23}.
For this reason, we consider the following assumption, commonly accepted in the literature.
\begin{enumerate}
\item We assume sufficiently large transmission distances such that the currents at the transmitter/transmitter/RISs are independent of the currents at the RISs/receiver/receiver, respectively, also known as the unilateral approximation \cite{ivr10}. This allows us to neglect the feedback links, i.e., we can consider $\mathbf{Z}_{TI}=\mathbf{0}$, $\mathbf{Z}_{TR}=\mathbf{0}$, and $\mathbf{Z}_{IR}=\mathbf{0}$.
\end{enumerate}
With this assumption, following the steps in \cite{gra21,ner23}, it is possible to show that the channel $\mathbf{H}$ is given by
\begin{multline}
\mathbf{H}=Z_0\left(Z_0\mathbf{I}+\mathbf{Z}_{RR}\right)^{-1}\\
\times\left(\mathbf{Z}_{RT}-\mathbf{Z}_{RI}\left(\mathbf{Z}_I+\mathbf{Z}_{II}\right)^{-1}\mathbf{Z}_{IT}\right)\mathbf{Z}_{TT}^{-1},\label{eq:HZ1}
\end{multline}
where $\mathbf{Z}_I\in\mathbb{C}^{LN_I\times LN_I}$ is a block diagonal matrix defined as
\begin{equation}
\mathbf{Z}_{I}=\text{diag}\left(\mathbf{Z}_{I,1},\mathbf{Z}_{I,2},\ldots,\mathbf{Z}_{I,L}\right),
\end{equation}
having in the $\ell$th block the impedance matrix of the reconfigurable impedance network at the $\ell$th RIS \cite{gra21,ner23}.

Since the impact of the reconfigurable impedance matrices of the RISs $\mathbf{Z}_{I,\ell}$ is not apparent in \eqref{eq:HZ1} due to the matrix inversion operation, it is necessary to further simplify this model.
To this end, we consider with no loss of generality that the signal sent by the transmitter reaches the receiver by flowing from the $1$st RIS to the $L$th RIS.
Thus, we make the following two additional assumptions to obtain an interpretable channel model.
\begin{enumerate}
\setcounter{enumi}{1}
\item We assume large enough distances between the RISs such that the currents at the $i$th RIS are independent of the currents at the $j$th RIS, with $i<j$.
As in assumption~1), this is referred to as the unilateral approximation and allows us to set to zero the feedback channels between the RISs, i.e., $\mathbf{Z}_{i,j}=\mathbf{0}$, $\forall i<j$ \cite{ivr10}.
\item We assume that the $\ell$th RIS is only connected to the transmitter, the receiver, the $(\ell-1)$th RIS (if $\ell\neq1$), and the $(\ell+1)$th RIS (if $\ell\neq L$), for $\ell=1,\ldots,L$.
In other words, we assume that the channel between the $i$th and $j$th RIS is completely obstructed when $i-j\geq2$, i.e., $\mathbf{Z}_{i,j}=\mathbf{0}$ if $i-j\geq2$.
This allows us to study the multi-hop cascaded channel formed by the $L$ RISs, as in related works \cite{han20}-\cite{ngu23}.
\end{enumerate}
Note that assumption 2) is needed to ensure that the currents at the $\ell$th RIS do not impact the currents at the previous RIS in the cascade, i.e., the $(\ell-1)$th RIS, for $\ell=2,\ldots,L$.
This assumption, named unilateral approximation, is also commonly made in \gls{mimo} system, where the currents at the receiver are assumed not to impact the currents at the transmitter \cite{ivr10}, and it is satisfied in practical wireless systems.
Besides, the rationale behind assumption 3) is that a routing algorithm has already selected $L$ RISs to help the communication between the transmitter and receiver, out of all the RISs available in the environment, e.g., as proposed in \cite{mei21,mei22-1}.
For more information on where to deploy the RISs in a multi-RIS environment, the reader is referred to \cite{li23-2}.

With the two additional assumptions 2) and 3), the matrix $\mathbf{Z}_{II}$ simplifies as
\begin{equation}
\mathbf{Z}_{II}=
\begin{bmatrix}
\mathbf{Z}_{II,1} &  &  &  & \mathbf{0}\\
\mathbf{Z}_{2,1} & \mathbf{Z}_{II,2} &  &  & \\
 & \mathbf{Z}_{3,2} & \ddots &  & \\
 &  & \ddots & \mathbf{Z}_{II,L-1} & \\
\mathbf{0} &  &  & \mathbf{Z}_{L,L-1} & \mathbf{Z}_{II,L}\\
\end{bmatrix},\label{eq:ZII}
\end{equation}
having non-zero block matrices only in the diagonal and subdiagonal.
Note that the blocks in the supradiagonal of $\mathbf{Z}_{II}$ are set to zero following assumption 2).
To simplify the channel model in \eqref{eq:HZ1} accordingly, we introduce the following proposition.

\begin{proposition}
Consider a square block matrix $\mathbf{M}\in\mathbb{C}^{LN\times LN}$ having square matrices $\mathbf{D}_\ell\in\mathbb{C}^{N\times N}$ in the diagonal, for $\ell=1,\ldots,L$, and square matrices $\mathbf{S}_{\ell,\ell-1}\in\mathbb{C}^{N\times N}$ in the subdiagonal, for $\ell=2,\ldots,L$, with all other blocks being zero matrices, i.e.,
\begin{equation}
\mathbf{M}=
\begin{bmatrix}
\mathbf{D}_1 &  &  &  & \mathbf{0}\\
\mathbf{S}_{2,1} & \mathbf{D}_{2} &  &  & \\
 & \mathbf{S}_{3,2} & \ddots &  & \\
 &  & \ddots & \mathbf{D}_{L-1} & \\
\mathbf{0} &  &  & \mathbf{S}_{L,L-1} & \mathbf{D}_L\\
\end{bmatrix}.\label{eq:M}
\end{equation}
If all $\mathbf{D}_\ell$ are invertible, the inverse of $\mathbf{M}$, denoted as $\mathbf{N}=\mathbf{M}^{-1}\in\mathbb{C}^{LN\times LN}$, is a square block matrix partitioned as
\begin{equation}
\mathbf{N}=
\begin{bmatrix}
\mathbf{N}_{1,1} & \cdots & \mathbf{N}_{1,L}\\
\vdots & \ddots & \vdots\\
\mathbf{N}_{L,1} & \cdots & \mathbf{N}_{L,L}
\end{bmatrix},
\end{equation}
where $\mathbf{N}_{i,j}\in\mathbb{C}^{N\times N}$ is given by\footnote{Note that the index $k$ in the product in \eqref{eq:N} decreases from $i-1$ to $j$ since $i>j$.
The decreasing order matters because of the non-commutativity of matrix multiplication.}
\begin{equation}
\mathbf{N}_{i,j}=
\begin{cases}
\mathbf{0} & \text{if }i<j\\
\mathbf{D}_{i}^{-1} & \text{if }i=j\\
\left(-1\right)^{i-j}\mathbf{D}_{i}^{-1}
\prod_{k=i-1}^{j}\left(\mathbf{S}_{k+1,k}\mathbf{D}_{k}^{-1}\right) & \text{if }i>j\\
\end{cases}.\label{eq:N}
\end{equation}
\label{pro}
\end{proposition}
\begin{proof}
To prove the proposition, it is sufficient to show that the matrix product $\mathbf{P}=\mathbf{M}\mathbf{N}\in\mathbb{C}^{LN\times LN}$, partitioned as
\begin{equation}
\mathbf{P}=
\begin{bmatrix}
\mathbf{P}_{1,1} & \cdots & \mathbf{P}_{1,L}\\
\vdots & \ddots & \vdots\\
\mathbf{P}_{L,1} & \cdots & \mathbf{P}_{L,L}
\end{bmatrix},
\end{equation}
with $\mathbf{P}_{i,j}\in\mathbb{C}^{N\times N}$, is the identity matrix, i.e., $\mathbf{P}_{i,j}=\mathbf{I}$, for $i=j$, and $\mathbf{P}_{i,j}=\mathbf{0}$, for $i\neq j$.
This can be directly shown by noticing that
\begin{equation}
\mathbf{P}_{1,j}=\mathbf{D}_{1}\mathbf{N}_{1,j},
\end{equation}
for $j=1,\dots,L$, and
\begin{equation}
\mathbf{P}_{i,j}=\mathbf{S}_{i,i-1}\mathbf{N}_{i-1,j}+\mathbf{D}_{i}\mathbf{N}_{i,j},
\end{equation}
for $i=2,\dots,L$ and $j=1,\dots,L$, according to the definition of $\mathbf{M}$ in \eqref{eq:M}, and by using \eqref{eq:N}.
\end{proof}

To rewrite the channel model in \eqref{eq:HZ1} and highlight the role of the impedance matrices $\mathbf{Z}_{I,\ell}$, we introduce the block matrix
\begin{equation}
\bar{\mathbf{Y}}=(\mathbf{Z}_I+\mathbf{Z}_{II})^{-1},
\end{equation}
partitioned as
\begin{equation}
\bar{\mathbf{Y}}=
\begin{bmatrix}
\bar{\mathbf{Y}}_{1,1} & \cdots & \bar{\mathbf{Y}}_{1,L}\\
\vdots & \ddots & \vdots\\
\bar{\mathbf{Y}}_{L,1} & \cdots & \bar{\mathbf{Y}}_{L,L}
\end{bmatrix},
\end{equation}
where $\bar{\mathbf{Y}}_{i,j}\in\mathbb{C}^{N_I\times N_I}$, for $i,j=1,\ldots,L$.
Since $\bar{\mathbf{Y}}$ is a block lower triangular matrix according to Proposition~\ref{pro}, \eqref{eq:HZ1} can be rewritten as
\begin{multline}
\mathbf{H}
=Z_0\left(Z_0\mathbf{I}+\mathbf{Z}_{RR}\right)^{-1}\\
\times\left(\mathbf{Z}_{RT}
-\sum_{\ell=1}^L\mathbf{Z}_{RI,\ell}\bar{\mathbf{Y}}_{\ell,\ell}\mathbf{Z}_{IT,\ell}\right.\\
\left.-\sum_{\ell=2}^L\mathbf{Z}_{RI,\ell}\sum_{k=1}^{\ell-1}\bar{\mathbf{Y}}_{\ell,k}\mathbf{Z}_{IT,k}
\right)\mathbf{Z}_{TT}^{-1},\label{eq:HZ2}
\end{multline}
where the role of the transmission impedance matrices $\mathbf{Z}_{RI,\ell}$ and $\mathbf{Z}_{IT,\ell}$ is highlighted.
In addition, since Proposition~\ref{pro} yields
\begin{equation}
\bar{\mathbf{Y}}_{\ell,\ell}=(\mathbf{Z}_{I,\ell}+\mathbf{Z}_{II,\ell})^{-1}\label{eq:pro1}
\end{equation}
and
\begin{multline}
\bar{\mathbf{Y}}_{\ell,k}=\left(-1\right)^{\ell-k}
\left(\mathbf{Z}_{I,\ell}+\mathbf{Z}_{II,\ell}\right)^{-1}\\
\times\prod_{p=\ell-1}^{k}\left(\mathbf{Z}_{p+1,p}\left(\mathbf{Z}_{I,p}+\mathbf{Z}_{II,p}\right)^{-1}\right),\label{eq:pro2}
\end{multline}
for $\ell>k$, the channel in \eqref{eq:HZ2} can be further rewritten as
\begin{multline}
\mathbf{H}
=Z_0\left(Z_0\mathbf{I}+\mathbf{Z}_{RR}\right)^{-1}\\
\times\left(\mathbf{Z}_{RT}
-\sum_{\ell=1}^L\mathbf{Z}_{RI,\ell}\left(\mathbf{Z}_{I,\ell}+\mathbf{Z}_{II,\ell}\right)^{-1}\mathbf{Z}_{IT,\ell}\right.\\
-\sum_{\ell=2}^L\mathbf{Z}_{RI,\ell}\sum_{k=1}^{\ell-1}
\left(-1\right)^{\ell-k}\left(\mathbf{Z}_{I,\ell}+\mathbf{Z}_{II,\ell}\right)^{-1}\\
\left.\times\prod_{p=\ell-1}^{k}\left(\mathbf{Z}_{p+1,p}\left(\mathbf{Z}_{I,p}+\mathbf{Z}_{II,p}\right)^{-1}\right)\mathbf{Z}_{IT,k}
\right)\mathbf{Z}_{TT}^{-1},\label{eq:HZ3}
\end{multline}
explicitly emphasizing the impact of the transmission impedance matrices $\mathbf{Z}_{\ell+1,\ell}$ and the reconfigurable impedance matrices of the RISs $\mathbf{Z}_{I,\ell}$.
Remarkably, \eqref{eq:HZ3} gives the channel model of a multi-RIS aided system in the impedance parameters, or $Z$-parameters, accounting for impedance mismatching and mutual coupling, of interest for less explored cascaded RIS-aided systems.
Note that by considering a single-RIS wireless system, i.e., by setting $L=1$ in the channel model in \eqref{eq:HZ3}, the channel model boils down to
\begin{multline}
\mathbf{H}
=Z_0\left(Z_0\mathbf{I}+\mathbf{Z}_{RR}\right)^{-1}\\
\times\left(\mathbf{Z}_{RT}
-\mathbf{Z}_{RI,1}\left(\mathbf{Z}_{I,1}+\mathbf{Z}_{II,1}\right)^{-1}\mathbf{Z}_{IT,1}\right)\mathbf{Z}_{TT}^{-1},
\end{multline}
which perfectly agrees with the channel model adopted in previous literature on physics-compliant modeling of single-RIS systems \cite{gra21,nos24-1,nos24-2,li24,ner23}.

\section{Channel Model with Perfect Matching and\\No Mutual Coupling}
\label{sec:model-no-MC}

In \eqref{eq:HZ3}, we have derived a general channel model for multi-RIS aided system.
In this section, we obtain further insights into the role of the RISs by assuming perfect matching and no mutual coupling at the transmitter, receiver, and RISs.
In detail, we consider the following two additional assumptions.
\begin{enumerate}
\setcounter{enumi}{3}
\item We assume perfect matching to the reference impedance $Z_0$ and no mutual coupling at the transmitter and receiver, i.e., $\mathbf{Z}_{TT}=Z_0\mathbf{I}$ and $\mathbf{Z}_{RR}=Z_0\mathbf{I}$.
\item We assume perfect matching to $Z_0$ and no mutual coupling at all the $L$ RISs, i.e., $\mathbf{Z}_{II,\ell}=Z_0\mathbf{I}$, for $\ell=1,\ldots,L$, which can be achieved by implementing the RISs through large reflectarrays with half-wavelength spacing.
\end{enumerate}
Following assumptions~4) and 5), the channel model in \eqref{eq:HZ3} can be simplified as
\begin{multline}
\mathbf{H}
=\frac{1}{2Z_0}\left(\mathbf{Z}_{RT}
-\sum_{\ell=1}^L\mathbf{Z}_{RI,\ell}\left(\mathbf{Z}_{I,\ell}+Z_0\mathbf{I}\right)^{-1}\mathbf{Z}_{IT,\ell}\right.\\
-\sum_{\ell=2}^L\mathbf{Z}_{RI,\ell}\sum_{k=1}^{\ell-1}
\left(-1\right)^{\ell-k}\left(\mathbf{Z}_{I,\ell}+Z_0\mathbf{I}\right)^{-1}\\
\left.\times\prod_{p=\ell-1}^{k}\left(\mathbf{Z}_{p+1,p}\left(\mathbf{Z}_{I,p}+Z_0\mathbf{I}\right)^{-1}\right)\mathbf{Z}_{IT,k}
\right),\label{eq:HZ4}
\end{multline}
giving the channel model in the $Z$-parameters of a multi-RIS aided system with perfect matching and no mutual coupling.

In the related literature, a RIS is often characterized through its scattering matrix \cite{she20}, which is related to the impedance matrix through a specific mapping according to microwave network theory \cite[Chapter 4]{poz11}.
Specifically, the scattering matrix of the $\ell$th RIS, denoted as $\boldsymbol{\Theta}_{\ell}\in\mathbb{C}^{N_I\times N_I}$, is related to $\mathbf{Z}_{I,\ell}$ through
\begin{equation}
\boldsymbol{\Theta}_{\ell}=\left(\mathbf{Z}_{I,\ell}+Z_0\mathbf{I}\right)^{-1}\left(\mathbf{Z}_{I,\ell}-Z_0\mathbf{I}\right),\label{eq:S}
\end{equation}
as discussed in \cite[Chapter 4]{poz11}.
Thus, by making the following two substitutions in \eqref{eq:HZ4}, we can obtain the channel model in terms of the RIS scattering matrices $\boldsymbol{\Theta}_{\ell}$, typically used in the literature.
First, by expressing \eqref{eq:S} as
\begin{align}
\boldsymbol{\Theta}_{\ell}
=&\left(\mathbf{Z}_{I,\ell}+Z_0\mathbf{I}\right)^{-1}\mathbf{Z}_{I,\ell}-Z_0\left(\mathbf{Z}_{I,\ell}+Z_0\mathbf{I}\right)^{-1}\\
=&\left(\mathbf{Z}_{I,\ell}+Z_0\mathbf{I}\right)^{-1}\mathbf{Z}_{I,\ell}+Z_0\left(\mathbf{Z}_{I,\ell}+Z_0\mathbf{I}\right)^{-1}\\
&-Z_0\left(\mathbf{Z}_{I,\ell}+Z_0\mathbf{I}\right)^{-1}-Z_0\left(\mathbf{Z}_{I,\ell}+Z_0\mathbf{I}\right)^{-1}\\
=&\mathbf{I}-2Z_0(\mathbf{Z}_{I,\ell}+Z_0\mathbf{I})^{-1},
\end{align}
we notice that 
\begin{equation}
\left(\mathbf{Z}_{I,\ell}+Z_0\mathbf{I}\right)^{-1}=-\frac{1}{2Z_0}\left(\boldsymbol{\Theta}_{\ell}-\mathbf{I}\right).\label{eq:sub1}
\end{equation}
Second, we use the notation
\begin{equation}
\mathbf{H}_{RT}=\frac{\mathbf{Z}_{RT}}{2Z_0},\:
\mathbf{H}_{RI,\ell}=\frac{\mathbf{Z}_{RI,\ell}}{2Z_0},\:
\mathbf{H}_{IT,\ell}=\frac{\mathbf{Z}_{IT,\ell}}{2Z_0},\label{eq:sub2}
\end{equation}
for $\ell=1,\ldots,L$, and
\begin{equation}
\mathbf{H}_{\ell+1,\ell}=\frac{\mathbf{Z}_{\ell+1,\ell}}{2Z_0},\label{eq:sub3}
\end{equation}
for $\ell=1,\ldots,L-1$, as introduced in \cite{she20}.
By substituting \eqref{eq:sub1}, \eqref{eq:sub2}, and \eqref{eq:sub3} in \eqref{eq:HZ4}, we obtain
\begin{multline}
\mathbf{H}=\mathbf{H}_{RT}
+\sum_{\ell=1}^L\mathbf{H}_{RI,\ell}(\boldsymbol{\Theta}_{\ell}-\mathbf{I})\mathbf{H}_{IT,\ell}\\
+\sum_{\ell=2}^L\mathbf{H}_{RI,\ell}(\boldsymbol{\Theta}_{\ell}-\mathbf{I})\sum_{k=1}^{\ell-1}
\prod_{p=\ell-1}^{k}\left(\mathbf{H}_{p+1,p}(\boldsymbol{\Theta}_{p}-\mathbf{I})\right)\mathbf{H}_{IT,k},\label{eq:HS4}
\end{multline}
representing the channel model in the multi-RIS aided scenario illustrated in Fig.~\ref{fig:system}.
Note that there are $1+L(L+1)/2$ additive terms in \eqref{eq:HS4}, each corresponding to a path that the signal can follow to flow from the transmitter to the receiver.
Specifically, there are $\ell$ different paths reaching the receiver passing through the channel $\mathbf{H}_{RI,\ell}$, for $\ell=1,\ldots,L$, giving $L(L+1)/2$ paths, in addition to the direct path $\mathbf{H}_{RT}$.

\begin{figure}[t]
\centering
\includegraphics[width=0.48\textwidth]{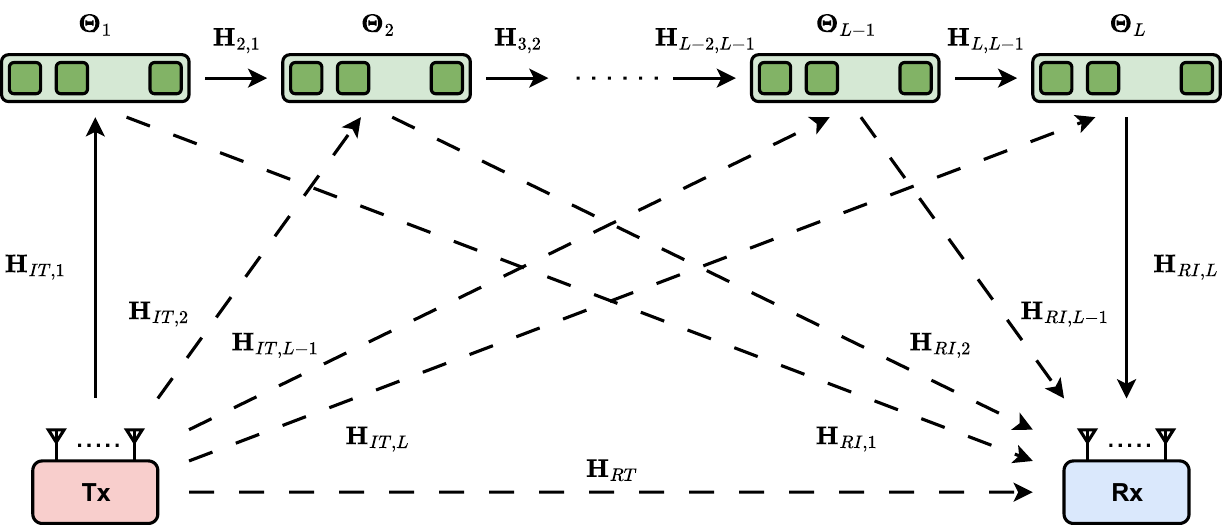}
\caption{Multi-RIS aided MIMO system.}
\label{fig:system}
\end{figure}

Numerous related studies on multi-RIS systems considered the case with $L=2$ RISs \cite{han20,zhe21-1,han22,an23,zhe21-2,qi23}.
For $L=2$, the physics-compliant model \eqref{eq:HS4} boils down to
\begin{multline}
\mathbf{H}=\mathbf{H}_{RT}
+\mathbf{H}_{RI,1}(\boldsymbol{\Theta}_{1}-\mathbf{I})\mathbf{H}_{IT,1}+\mathbf{H}_{RI,2}(\boldsymbol{\Theta}_{2}-\mathbf{I})\mathbf{H}_{IT,2}\\
+\mathbf{H}_{RI,2}(\boldsymbol{\Theta}_{2}-\mathbf{I})
\mathbf{H}_{2,1}(\boldsymbol{\Theta}_{1}-\mathbf{I})\mathbf{H}_{IT,1}.\label{eq:HL2}
\end{multline}
Notably, \eqref{eq:HL2} slightly differs from the channel model widely used in related literature, which is instead given by
\begin{multline}
\mathbf{H}^\prime=\mathbf{H}_{RT}
+\mathbf{H}_{RI,1}\boldsymbol{\Theta}_{1}\mathbf{H}_{IT,1}+\mathbf{H}_{RI,2}\boldsymbol{\Theta}_{2}\mathbf{H}_{IT,2}\\
+\mathbf{H}_{RI,2}\boldsymbol{\Theta}_{2}\mathbf{H}_{2,1}\boldsymbol{\Theta}_{1}\mathbf{H}_{IT,1},\label{eq:HL2prime}
\end{multline}
as used in \cite{han20,zhe21-1,han22,an23,zhe21-2,qi23}.
Interestingly, we observe that the only difference between \eqref{eq:HL2} and \eqref{eq:HL2prime} lies in the fact that the terms $(\boldsymbol{\Theta}_{1}-\mathbf{I})$ and $(\boldsymbol{\Theta}_{2}-\mathbf{I})$ in \eqref{eq:HL2} are replaced by $\boldsymbol{\Theta}_{1}$ and $\boldsymbol{\Theta}_{2}$ in \eqref{eq:HL2prime}, respectively.
This is because of the structural scattering effects of RIS, resulting in a specular reflection, commonly neglected in the communication theoretic literature, as observed in \cite{nos24-1,nos24-2,abr23}.

Since the channel model in \eqref{eq:HS4} includes the effect of wireless channels that may be fully obstructed in practice, we discuss a scenario with practical interest in the following, assuming that the transmitter and receiver are connected only to RIS $1$ and $L$, respectively.
In other words, we make the following assumption.
\begin{enumerate}
\item[6)] We assume that the link between the transmitter and RIS $2,\ldots,L$ is fully blocked, i.e., $\mathbf{Z}_{IT,\ell}=\mathbf{0}$, for $\ell=2,\ldots,L$, the link between the receiver and RIS $1,\ldots,L-1$ is blocked, i.e., $\mathbf{Z}_{RI,\ell}=\mathbf{0}$, for $\ell=1,\ldots,L-1$, the link between the transmitter and receiver is blocked, i.e., $\mathbf{Z}_{RT}=\mathbf{0}$.
\end{enumerate}
With this simplifying assumption, \eqref{eq:HZ4} gives
\begin{multline}
\mathbf{H}
=-\frac{\left(-1\right)^{L-1}}{2Z_0}\mathbf{Z}_{RI,L}
\left(\mathbf{Z}_{I,L}+Z_0\mathbf{I}\right)^{-1}\\
\times\prod_{\ell=L-1}^{1}\left(\mathbf{Z}_{\ell+1,\ell}\left(\mathbf{Z}_{I,\ell}+Z_0\mathbf{I}\right)^{-1}\right)\mathbf{Z}_{IT,1},\label{eq:HZ5-C}
\end{multline}
while \eqref{eq:HS4} gives
\begin{equation}
\mathbf{H}=
\mathbf{H}_{RI,L}\left(\boldsymbol{\Theta}_{L}-\mathbf{I}\right)
\prod_{\ell=L-1}^{1}\left(\mathbf{H}_{\ell+1,\ell}\left(\boldsymbol{\Theta}_{\ell}-\mathbf{I}\right)\right)\mathbf{H}_{IT,1}.\label{eq:H}
\end{equation}
This communication scenario has been commonly studied in related literature on multi-RIS aided communications with a generic number of RISs $L$ \cite{mei21,hua21,mei22-2,mei22-1,ma22,ngu23}.
However, the channel model widely used in related literature slightly differs from \eqref{eq:H}, and is given by
\begin{equation}
\mathbf{H}^\prime=\mathbf{H}_{RI,L}\boldsymbol{\Theta}_{L}\prod_{\ell=L-1}^{1}\left(\mathbf{H}_{\ell+1,\ell}\boldsymbol{\Theta}_{\ell}\right)\mathbf{H}_{IT,1}.\label{eq:Hprime}
\end{equation}
Note that the difference between \eqref{eq:H} and \eqref{eq:Hprime} is that the terms $(\boldsymbol{\Theta}_{\ell}-\mathbf{I})$ in \eqref{eq:H} are replaced by the terms $\boldsymbol{\Theta}_{\ell}$ in \eqref{eq:Hprime}, for $\ell=1,\ldots,L$.
Remarkably, this is due to the fact that the RIS structural scattering is omitted in the channel model \eqref{eq:Hprime} used in previous literature \cite{nos24-1,nos24-2,abr23}.

\section{Channel Model in the Presence of\\Reflective, Hybrid, and Multi-Sector RISs}
\label{sec:model-multi-sector}

The channel models developed in Sections~\ref{sec:model} and \ref{sec:model-no-MC} implicitly assume that all the RISs in the system are reflective RISs.
In this section, we generalize our physics-compliant channel model under assumptions~1) to 6), to multi-RIS aided systems where the RISs can be reflective, hybrid, or multi-sector RISs.

We recall that reflective RISs can only reflect the impinging signal within a half-space region, since the $N_I$ elements are all placed on one side of the surface.
Besides, hybrid RISs, also known as \glspl{star-ris}, are composed of unit cells, each made of 2 RIS elements placed back-to-back in the opposite sides of the surface \cite{li22-1}.
Thus, the impinging signal can be reflected within the same half-space region, or transmitted through the RIS to reach a receiver in the other half-space region.
Following the model in \cite{li22-1}, in hybrid RISs the $N_I$ elements are divided into the two sides of the planar surface, each side covering half-space and including $N_I/2$ elements.
Similarly, in multi-sector RISs with $S$ sectors, the $N_I$ elements are divided into $S$ sectors (the sides of a prism), each covering $1/S$ of the space and including $N_S=N_I/S$ elements \cite{li22-2}.
Note that hybrid RISs are a special case of multi-sector RISs, i.e., with $S=2$ \cite{li22-2}.

To model multi-RIS aided systems containing reflective, hybrid, and multi-sector RISs, we mathematically describe all the RISs as multi-sector RISs, as this description includes reflective and hybrid RISs as special cases.
Specifically, we describe the $\ell$th RIS as a multi-sector RIS with $S_\ell$ sectors, each including $N_{S_\ell}=N_I/S_\ell$ elements, for $\ell=1,\ldots,L$.
Thus, to model the $\ell$th RIS as a reflective RIS or a hybrid RIS, we can set $S_\ell=1$ or $S_\ell=2$, respectively.
For better clarity, we study the impact of the presence of reflective, hybrid, and multi-sector RISs on the simplified channel in \eqref{eq:H}, i.e., derived under the assumptions~1) to 6).
In this case, a RIS receives the signal from only one device (the transmitter or the previous RIS) and forwards it to only one device (the receiver or the following RIS).
Thus, any hybrid or multi-sector RIS is used either in reflective or transmissive mode and all the impinging power is either reflected or transmitted, depending on the operating mode of the RIS.

With all the RISs being multi-sector RISs, the channels $\mathbf{H}_{IT,1}$, $\mathbf{H}_{RI,L}$, and $\mathbf{H}_{\ell+1,\ell}$, and the terms $(\boldsymbol{\Theta}_{\ell}-\mathbf{I})$ in \eqref{eq:H} can be conveniently partitioned as follows.
The channel between the transmitter and the $1$st RIS $\mathbf{H}_{IT,1}$ can be expressed as
\begin{equation}
\mathbf{H}_{IT,1}=\left[\mathbf{H}_{IT,1}^{(1)T},\mathbf{H}_{IT,1}^{(2)T},\ldots,\mathbf{H}_{IT,1}^{(S_1)T}\right]^T,\label{eq:HIT-MS}
\end{equation}
where $\mathbf{H}_{IT,1}^{(s)}\in\mathbb{C}^{N_{S_1}\times N_T}$ is the channel from the transmitter to the $s$th sector of the $1$st RIS, for $s=1,\ldots,S_1$, and $S_1$ is the number of sectors of the $1$st RIS.
Note that if the $1$st RIS is a reflective RIS, we trivially have $S_1=1$ and $\mathbf{H}_{IT,1}=\mathbf{H}_{IT,1}^{(1)}$.
The channel between $L$th RIS and the receiver $\mathbf{H}_{RI,L}$ can be given by
\begin{equation}
\mathbf{H}_{RI,L}=\left[\mathbf{H}_{RI,L}^{(1)},\mathbf{H}_{RI,L}^{(2)},\ldots,\mathbf{H}_{RI,L}^{(S_L)}\right],\label{eq:HRI-MS}
\end{equation}
where $\mathbf{H}_{RI,L}^{(s)}\in\mathbb{C}^{N_R\times N_{S_L}}$ is the channel from the $s$th sector of the $L$th RIS to the receiver, for $s=1,\ldots,S_L$, and $S_L$ is the number of sectors of the $L$th RIS.
The channel between $\ell$th RIS and the $(\ell+1)$th RIS $\mathbf{H}_{\ell+1,\ell}$ can be partitioned as
\begin{equation}
\mathbf{H}_{\ell+1,\ell}=
\begin{bmatrix}
\mathbf{H}_{\ell+1,\ell}^{(1,1)} & \mathbf{H}_{\ell+1,\ell}^{(1,2)} & \cdots & \mathbf{H}_{\ell+1,\ell}^{(1,S_{\ell})}\\
\mathbf{H}_{\ell+1,\ell}^{(2,1)} & \mathbf{H}_{\ell+1,\ell}^{(2,2)} & \cdots & \mathbf{H}_{\ell+1,\ell}^{(2,S_{\ell})}\\
\vdots & \vdots & \ddots & \vdots\\
\mathbf{H}_{\ell+1,\ell}^{(S_{\ell+1},1)} & \mathbf{H}_{\ell+1,\ell}^{(S_{\ell+1},2)} & \cdots & \mathbf{H}_{\ell+1,\ell}^{(S_{\ell+1},S_{\ell})}
\end{bmatrix},\label{eq:HL-MS}
\end{equation}
with $\mathbf{H}_{\ell+1,\ell}^{(t,s)}\in\mathbb{C}^{N_{S_{\ell+1}}\times N_{S_{\ell}}}$ being the channel from the $s$th sector of the $\ell$th RIS to the $t$th sector of the $(\ell+1)$th RIS, for $s=1,\ldots,S_\ell$ and $t=1,\ldots,S_{\ell+1}$.
Similarly, the terms $(\boldsymbol{\Theta}_{\ell}-\mathbf{I})$ can be partitioned as
\begin{equation}
\boldsymbol{\Theta}_{\ell}-\mathbf{I}=
\begin{bmatrix}
\boldsymbol{\Theta}_{\ell}^{(1,1)}-\mathbf{I} & \boldsymbol{\Theta}_{\ell}^{(1,2)} & \cdots & \boldsymbol{\Theta}_{\ell}^{(1,S_{\ell})}\\
\boldsymbol{\Theta}_{\ell}^{(2,1)} & \boldsymbol{\Theta}_{\ell}^{(2,2)}-\mathbf{I} & \cdots & \boldsymbol{\Theta}_{\ell}^{(2,S_{\ell})}\\
\vdots & \vdots & \ddots & \vdots\\
\boldsymbol{\Theta}_{\ell}^{(S_{\ell},1)} & \boldsymbol{\Theta}_{\ell}^{(S_{\ell},2)} & \cdots & \boldsymbol{\Theta}_{\ell}^{(S_{\ell},S_{\ell})}-\mathbf{I}
\end{bmatrix},\label{eq:T-MS}
\end{equation}
with $\boldsymbol{\Theta}_{\ell}^{(s,s)}\in\mathbb{C}^{N_{S_{\ell}}\times N_{S_{\ell}}}$ representing the reflection scattering matrix at the $s$th sector of the $\ell$th RIS, and $\boldsymbol{\Theta}_{\ell}^{(t,s)}\in\mathbb{C}^{N_{S_{\ell}}\times N_{S_{\ell}}}$ being the transmission scattering matrix from the $s$th sector to the $t$th sector of the $\ell$th RIS, for $s,t=1,\ldots,S_\ell$.

To simplify the channels \eqref{eq:HIT-MS}, \eqref{eq:HRI-MS}, and \eqref{eq:HL-MS} with the model of multi-secor RIS developed in \cite{li22-2}, we denote as $s_{\ell,A}$ and $s_{\ell,D}$ the sector of the $\ell$th RIS where the signal arrives and departs, respectively, for $\ell=1,\ldots,L$.
Note that, if the $\ell$th RIS is a reflective RIS, i.e., $S_\ell=1$, we trivially have $s_{\ell,A}=s_{\ell,D}=1$.
Thus, according to \cite{li22-2}, we have $\mathbf{H}_{IT,1}^{(s)}=\mathbf{0}$ if $s\neq s_{1,A}$, $\mathbf{H}_{RI,L}^{(s)}=\mathbf{0}$ if $s\neq s_{L,D}$, and $\mathbf{H}_{\ell+1,\ell}^{(t,s)}=\mathbf{0}$ if $s\neq s_{\ell,D}$ or $t\neq s_{\ell+1,A}$.
With these considerations, we introduce $\bar{\mathbf{H}}_{IT,1}=\mathbf{H}_{IT,1}^{(s_{1,A})}$, $\bar{\mathbf{H}}_{RI,L}=\mathbf{H}_{RI,L}^{(s_{L,D})}$, $\bar{\mathbf{H}}_{\ell+1,\ell}=\mathbf{H}_{\ell+1,\ell}^{(s_{\ell+1,A},s_{\ell,D})}$, and $\bar{\boldsymbol{\Theta}}_{\ell}=\boldsymbol{\Theta}_{\ell}^{(s_{\ell,D},s_{\ell,A})}$, and we can compactly rewrite the channel model in \eqref{eq:H} for systems with multi-sector RISs as
\begin{multline}
\mathbf{H}=
\bar{\mathbf{H}}_{RI,L}\left(\bar{\boldsymbol{\Theta}}_{L}-\delta_{s_{L,D},s_{L,A}}\mathbf{I}\right)\\
\times\prod_{\ell=L-1}^{1}\left(\bar{\mathbf{H}}_{\ell+1,\ell}\left(\bar{\boldsymbol{\Theta}}_{\ell}-\delta_{s_{\ell,D},s_{\ell,A}}\mathbf{I}\right)\right)\bar{\mathbf{H}}_{IT,1},\label{eq:H-MS}
\end{multline}
where $\delta_{t,s}$ represents the Kronecker delta, being $\delta_{t,s}=1$ if $t=s$ and $\delta_{t,s}=0$ if $t\neq s$.
Remarkably, the term $(\bar{\boldsymbol{\Theta}}_{\ell}-\delta_{s_{\ell,D},s_{\ell,A}}\mathbf{I})$ becomes $(\bar{\boldsymbol{\Theta}}_{\ell}-\mathbf{I})$ when the $\ell$th RIS is a reflective, hybrid, or multi-sector RIS used in reflective mode, i.e., the signal arrives and departs from the same sector.
Conversely, the term $(\bar{\boldsymbol{\Theta}}_{\ell}-\delta_{s_{\ell,D},s_{\ell,A}}\mathbf{I})$ boils down to $\bar{\boldsymbol{\Theta}}_{\ell}$ if the $\ell$th RIS is a hybrid or multi-sector RIS used in transmissive mode, i.e., the signal arrives and departs from two different sectors.
Thus, we observe that the structural scattering of RISs used in transmissive mode does not impact the channel.
This occurs because the effect of the structural scattering consists in a specular reflection of the \gls{em} signal at the RIS \cite{abr23}.
Consequently, this specular reflection cannot reach the receiver if it is located in a different sector than the transmitter, assuming that receiver and transmitter are coupled only with the RIS elements of their respective sectors, and RIS elements in different sectors are not coupled with each other.

As in the case of reflective RISs, we notice that the obtained channel model in \eqref{eq:H-MS} differs from the channel model for multi-RIS systems including transmissive RISs used in previous literature, given by
\begin{equation}
\mathbf{H}^\prime=\bar{\mathbf{H}}_{RI,L}\bar{\boldsymbol{\Theta}}_{L}\prod_{\ell=L-1}^{1}\left(\bar{\mathbf{H}}_{\ell+1,\ell}\bar{\boldsymbol{\Theta}}_{\ell}\right)\bar{\mathbf{H}}_{IT,1},\label{eq:Hprime-MS}
\end{equation}
as adopted in \glspl{star-ris} literature \cite{jav24}.
Interestingly, the channel models in \eqref{eq:H-MS} and \eqref{eq:Hprime-MS} coincide when all the RISs are hybrid or multi-sector RISs, and are used in transmissive mode, i.e., $s_{\ell,A}\neq s_{\ell,D}$, for $\ell=1,\ldots,L$, since the structural scattering of RISs used in transmissive mode does not alter the channel.

\section{Channel Gain Scaling Laws with\\Line-of-Sight Channels}
\label{sec:gain-los}

We have derived the channel expression of multi-RIS aided systems, and noticed that it differs from the widely used channel expression in the presence of reflective RISs.
In this section, we compare the physics-compliant and widely used models by deriving the scaling laws of their channel gains under \gls{los} channels, i.e., how they scale with $N_I$ as $N_I$ grows large.
Since the physics-compliant model differs from the widely used one in the presence of reflective RISs, we assume all RISs to work in reflective mode in the following.

Considering the physics-compliant channel model in \eqref{eq:H} and the widely used model in \eqref{eq:Hprime} under \gls{los} channels, the channel between the $L$th RIS and the receiver writes as $\mathbf{H}_{RI,L}=\Lambda_{RI,L}\mathbf{a}_{RI,L}\mathbf{b}_{RI,L}^T$, with $\Lambda_{RI,L}\in\mathbb{R}$ being the path gain, $\mathbf{a}_{RI,L}=[e^{j\alpha_{RI,L,1}},\ldots,e^{j\alpha_{RI,L,N_R}}]^T$, and $\mathbf{b}_{RI,L}=[e^{j\beta_{RI,L,1}},\ldots,e^{j\beta_{RI,L,N_I}}]^T$.
Similarly, the channel between the transmitter and the $1$st RIS is given by $\mathbf{H}_{IT,1}=\Lambda_{IT,1}\mathbf{a}_{IT,1}\mathbf{b}_{IT,1}^T$, where $\Lambda_{IT,1}\in\mathbb{R}$ is the path gain, $\mathbf{a}_{IT,1}=[e^{j\alpha_{IT,1,1}},\ldots,e^{j\alpha_{IT,1,N_I}}]^T$, and $\mathbf{b}_{IT,1}=[e^{j\beta_{IT,1,1}},\ldots,e^{j\beta_{IT,1,N_T}}]^T$.
Besides, $\mathbf{H}_{\ell+1,\ell}=\Lambda_{\ell+1,\ell}\mathbf{a}_{\ell+1,\ell}\mathbf{b}_{\ell+1,\ell}^T$, with $\Lambda_{\ell+1,\ell}\in\mathbb{R}$ being the path gain, $\mathbf{a}_{\ell+1,\ell}=[e^{j\alpha_{\ell+1,\ell,1}},\ldots,e^{j\alpha_{\ell+1,\ell,N_I}}]^T$, and $\mathbf{b}_{\ell+1,\ell}=[e^{j\beta_{\ell+1,\ell,1}},\ldots,e^{j\beta_{\ell+1,\ell,N_I}}]^T$, for $\ell=1,\ldots,L-1$.
Considering the conventional D-RIS architecture (which is optimal under \gls{los} channels \cite{she20}) we have $\boldsymbol{\Theta}_{\ell}=\text{diag}(e^{j\theta_{\ell,1}},\ldots,e^{j\theta_{\ell,N_I}})$, with $\theta_{\ell,n_I}$ being the $n_I$th phase shift of the $\ell$th RIS, for $n_I=1,\ldots,N_I$ and $\ell=1,\ldots,L$.
In the following, we investigate and compare the maximum channel gain achievable in the case of the physics-compliant and the widely used model.

\subsection{Physics-Compliant Channel Model}

For the channel model in \eqref{eq:H}, we can derive a global optimal solution for the scattering matrices $\boldsymbol{\Theta}_{\ell}$ to maximize the channel gain $\Vert\mathbf{H}\Vert^2$ by adapting the optimization method proposed in \cite{mei21} for the widely used channel model.
Specifically, in the case of \gls{los} channels, \eqref{eq:H} can be expressed as
\begin{multline}
\mathbf{H}=
\Lambda
\mathbf{a}_{RI,L}\mathbf{b}_{RI,L}^T\left(\boldsymbol{\Theta}_{L}-\mathbf{I}\right)\mathbf{a}_{L,L-1}\\
\times\left(\prod_{\ell=L-1}^{2}\mathbf{b}_{\ell+1,\ell}^T\left(\boldsymbol{\Theta}_{\ell}-\mathbf{I}\right)\mathbf{a}_{\ell,\ell-1}\right)
\mathbf{b}_{2,1}^T\left(\boldsymbol{\Theta}_{1}-\mathbf{I}\right)\mathbf{a}_{IT,1}\mathbf{b}_{IT,1}^T,\label{eq:Hlos}
\end{multline}
where $\Lambda=\Lambda_{RI,L}\prod_{\ell=L-1}^{1}(\Lambda_{\ell+1,\ell})\Lambda_{IT,1}$ is the total path gain.
Interestingly, \eqref{eq:Hlos} can be rewritten as an outer product $\mathbf{a}_{RI,L}\mathbf{b}_{IT,1}^T$ multiplied by $L$ scalars $K_{\ell}\in\mathbb{C}$, i.e.,
\begin{equation}
\mathbf{H}=
\Lambda\prod_{\ell=L}^{1}\left(K_{\ell}\right)\mathbf{a}_{RI,L}\mathbf{b}_{IT,1}^T,
\end{equation}
where $K_{L}=\mathbf{b}_{RI,L}^T(\boldsymbol{\Theta}_{L}-\mathbf{I})\mathbf{a}_{L,L-1}$, $K_{1}=\mathbf{b}_{2,1}^T(\boldsymbol{\Theta}_{1}-\mathbf{I})\mathbf{a}_{IT,1}$, and $K_{\ell}=\mathbf{b}_{\ell+1,\ell}^T(\boldsymbol{\Theta}_{\ell}-\mathbf{I})\mathbf{a}_{\ell,\ell-1}$, for $\ell=2,\ldots,L-1$.
Noting that the channel gain, i.e., the spectral norm of $\mathbf{H}$, writes as
\begin{equation}
\left\Vert\mathbf{H}\right\Vert^2
=\Lambda^2\prod_{\ell=L}^{1}\left\vert K_{\ell}\right\vert^2\left\Vert\mathbf{a}_{RI,L}\mathbf{b}_{IT,1}^T\right\Vert^2,\label{eq:proof1}
\end{equation}
each $\boldsymbol{\Theta}_{\ell}=\text{diag}(\theta_{\ell,1},\ldots,\theta_{\ell,N_I})$ can be individually globally optimized to maximize $\vert K_{\ell}\vert^2$, respectively, by setting
\begin{multline}
\theta_{L,n_I}=-\arg(\mathbf{b}_{RI,L}^T\mathbf{a}_{L,L-1})\\
-\arg([\mathbf{b}_{RI,L}]_{n_I})-\arg([\mathbf{a}_{L,L-1}]_{n_I}),\label{eq:TL}
\end{multline}
\begin{multline}
\theta_{1,n_I}=-\arg(\mathbf{b}_{2,1}^T\mathbf{a}_{IT,1})\\
-\arg([\mathbf{b}_{2,1}]_{n_I})-\arg([\mathbf{a}_{IT,1}]_{n_I}),\label{eq:T1}
\end{multline}
\begin{multline}
\theta_{\ell,n_I}=-\arg(\mathbf{b}_{\ell+1,\ell}^T\mathbf{a}_{\ell,\ell-1})\\
-\arg([\mathbf{b}_{\ell+1,\ell}]_{n_I})-\arg([\mathbf{a}_{\ell,\ell-1}]_{n_I}),\label{eq:Tell}
\end{multline}
for $\ell=2,\ldots,L-1$ and $n_I=1,\ldots,N_I$, giving
$\vert K_{L}\vert^2=(\vert\mathbf{b}_{RI,L}^T\mathbf{a}_{L,L-1}\vert+N_I)^2$,
$\vert K_{1}\vert^2=(\vert\mathbf{b}_{2,1}^T\mathbf{a}_{IT,1}\vert+N_I)^2$, and
$\vert K_{\ell}\vert^2=(\vert\mathbf{b}_{\ell+1,\ell}^T\mathbf{a}_{\ell,\ell-1}\vert+N_I)^2$, for $\ell=2,\ldots,L-1$.
By substituting the optimal values of $\vert K_{\ell}\vert^2$ into \eqref{eq:proof1}, the optimal channel gain is given by
\begin{multline}
\left\Vert\mathbf{H}\right\Vert^2
=\Lambda^2\left(\left\vert\mathbf{b}_{RI,L}^T\mathbf{a}_{L,L-1}\right\vert+N_I\right)^2\\
\times\prod_{\ell=L-1}^{2}\left(\left\vert\mathbf{b}_{\ell+1,\ell}^T\mathbf{a}_{\ell,\ell-1}\right\vert+N_I\right)^2\\
\times\left(\left\vert\mathbf{b}_{2,1}^T\mathbf{a}_{IT,1}\right\vert+N_I\right)^2\left\Vert\mathbf{a}_{RI,L}\mathbf{b}_{IT,1}^T\right\Vert^2.\label{eq:proof2}
\end{multline}
In addition, since $\Vert\mathbf{a}_{RI,L}\mathbf{b}_{IT,1}^T\Vert^2=\Vert\mathbf{a}_{RI,L}\Vert^2\Vert\mathbf{b}_{IT,1}\Vert^2$ (because of the property of the spectral norm of an outer product) and $\Vert\mathbf{a}_{RI,L}\Vert^2=N_R$ and $\Vert\mathbf{b}_{IT,1}\Vert^2=N_T$, we obtain
\begin{multline}
\left\Vert\mathbf{H}\right\Vert^2=\Lambda^2\left(\left\vert\mathbf{b}_{RI,L}^T\mathbf{a}_{L,L-1}\right\vert+N_I\right)^2\\
\times\prod_{\ell=L-1}^{2}\left(\left\vert\mathbf{b}_{\ell+1,\ell}^T\mathbf{a}_{\ell,\ell-1}\right\vert+N_I\right)^2\left(\left\vert\mathbf{b}_{2,1}^T\mathbf{a}_{IT,1}\right\vert+N_I\right)^2N_RN_T,
\end{multline}
which depends on the specific channel realizations $\mathbf{H}_{RI,L}$, $\mathbf{H}_{IT,1}$, and $\mathbf{H}_{\ell+1,\ell}$, for $\ell=1,\ldots,L-1$.

To derive the expected value $\text{E}[\Vert\mathbf{H}\Vert^2]$, we consider the entries of the vectors $\mathbf{a}_{IT,1}$, $\mathbf{b}_{RI,L}$, $\mathbf{a}_{\ell+1,\ell}$, and $\mathbf{b}_{\ell+1,\ell}$ to be \gls{iid} with unit modulus and phase uniformly distributed in $[0,2\pi)$.
In this way, the three scalar products $\mathbf{b}_{RI,L}^T\mathbf{a}_{L,L-1}=\sum_{n_I=1}^{N_I}[\mathbf{b}_{RI,L}]_{n_I}[\mathbf{a}_{L,L-1}]_{n_I}$,
$\mathbf{b}_{\ell+1,\ell}^T\mathbf{a}_{\ell,\ell-1}=\sum_{n_I=1}^{N_I}[\mathbf{b}_{\ell+1,\ell}]_{n_I}[\mathbf{a}_{\ell,\ell-1}]_{n_I}$, and
$\mathbf{b}_{2,1}^T\mathbf{a}_{IT,1}=\sum_{n_I=1}^{N_I}[\mathbf{b}_{2,1}]_{n_I}[\mathbf{a}_{IT,1}]_{n_I}$ are three independent sums of $N_I$ \gls{iid} complex random variables with unit modulus and phase uniformly distributed in $[0,2\pi)$, i.e., having mean $0$ and variance $1$.
Following the Central Limit Theorem, $\mathbf{b}_{RI,L}^T\mathbf{a}_{L,L-1}$, $\mathbf{b}_{\ell+1,\ell}^T\mathbf{a}_{\ell,\ell-1}$, and $\mathbf{b}_{2,1}^T\mathbf{a}_{IT,1}$ are therefore independent and all distributed as $\mathcal{CN}(0,N_I)$, for $N_I$ large enough.
Thus, we have
\begin{equation}
\text{E}\left[\left\Vert\mathbf{H}\right\Vert^2\right]=\Lambda^2\text{E}\left[\left(\left\vert c\right\vert+N_I\right)^2\right]^LN_RN_T,
\end{equation}
where $c\sim\mathcal{CN}(0,N_I)$.
By using $\text{E}[\vert c\vert]=\sqrt{\frac{\pi}{4}N_I}$ and $\text{E}[\vert c\vert^2]=N_I$, we obtain
\begin{equation}
\text{E}\left[\left\Vert\mathbf{H}\right\Vert^2\right]=\Lambda^2\left(N_I^2+\sqrt{\pi N_I}N_I+N_I\right)^LN_RN_T,\label{eq:EG}
\end{equation}
giving the scaling law (for sufficiently large $N_I$) of the average channel gain of the physics-compliant model.

\subsection{Widely Used Channel Model}

Considering the widely used channel model in \eqref{eq:Hprime} under \gls{los} channels, we have
\begin{equation}
\mathbf{H}^\prime=
\Lambda\prod_{\ell=L}^{1}\left(K_{\ell}^\prime\right)\mathbf{a}_{RI,L}\mathbf{b}_{IT,1}^T,
\end{equation}
where $K_{\ell}^\prime\in\mathbb{C}$ are given by $K_{L}^\prime=\mathbf{b}_{RI,L}^T\boldsymbol{\Theta}_{L}\mathbf{a}_{L,L-1}$, $K_{1}^\prime=\mathbf{b}_{2,1}^T\boldsymbol{\Theta}_{1}\mathbf{a}_{IT,1}$, and $K_{\ell}^\prime=\mathbf{b}_{\ell+1,\ell}^T\boldsymbol{\Theta}_{\ell}\mathbf{a}_{\ell,\ell-1}$, for $\ell=2,\ldots,L-1$.
As discussed in \cite{mei21}, the channel gain $\Vert\mathbf{H}^\prime\Vert^2$ can be globally maximized by setting
\begin{align}
\theta_{L,n_I}&=-\arg([\mathbf{b}_{RI,L}]_{n_I})-\arg([\mathbf{a}_{L,L-1}]_{n_I}),\label{eq:sol-wu1}\\
\theta_{1,n_I}&=-\arg([\mathbf{b}_{2,1}]_{n_I})-\arg([\mathbf{a}_{IT,1}]_{n_I}),\label{eq:sol-wu2}\\
\theta_{\ell,n_I}&=-\arg([\mathbf{b}_{\ell+1,\ell}]_{n_I})-\arg([\mathbf{a}_{\ell,\ell-1}]_{n_I}),\label{eq:sol-wu3}
\end{align}
for $\ell=2,\ldots,L-1$ and $n_I=1,\ldots,N_I$.
Accordingly, the maximum channel gain and its expected value are given by
\begin{equation}
\left\Vert\mathbf{H}^\prime\right\Vert^2=\text{E}\left[\left\Vert\mathbf{H}^\prime\right\Vert^2\right]=\Lambda^2N_I^{2L}N_RN_T,\label{eq:Gprime}
\end{equation}
in agreement with previous literature \cite{mei21}.


\section{Channel Gain Maximization with\\Multipath Channels}
\label{sec:gain-multipath}

In this section, we analyze the achievable channel gains for the physics-compliant and the widely used channel models, under generic multipath channels, i.e., we do not assume \gls{los} and do not make any prior assumption on the channel matrices.
Since closed-form expressions of their scaling laws are not available in this case, we propose an iterative algorithm to maximize the channel gains by optimizing the RIS scattering matrices.
In addition, we provide upper bounds on the channel gains to verify the effectiveness of the proposed optimization algorithm.

\subsection{Physics-Compliant Channel Model}

Consider the physics-compliant channel model in \eqref{eq:H} in the presence of multipath channels.
We propose to maximize the gain of this channel $\Vert\mathbf{H}\Vert^2$ by iteratively optimizing the $L$ RIS scattering matrices.
In detail, when the $\ell$th RIS scattering matrix $\boldsymbol{\Theta}_{\ell}$ is optimized with the other $L-1$ scattering matrices being fixed, the channel model in \eqref{eq:H} can be rewritten as
\begin{equation}
\mathbf{H}=\mathbf{H}_{RI,\ell}\left(\boldsymbol{\Theta}_{\ell}-\mathbf{I}\right)\mathbf{H}_{IT,\ell},
\end{equation}
where
\begin{equation}
\mathbf{H}_{RI,\ell}=\mathbf{H}_{RI,L}\prod_{k=L}^{\ell+1}\left(\left(\boldsymbol{\Theta}_{k}-\mathbf{I}\right)\mathbf{H}_{k,k-1}\right),\label{eq:hRIl}
\end{equation}
if $\ell=1,\ldots,L-1$, and
\begin{equation}
\mathbf{H}_{IT,\ell}=\prod_{k=\ell-1}^{1}\left(\mathbf{H}_{k+1,k}\left(\boldsymbol{\Theta}_{k}-\mathbf{I}\right)\right)\mathbf{H}_{IT,1},\label{eq:hITl}
\end{equation}
if $\ell=2,\ldots,L$.
Thus, $\boldsymbol{\Theta}_{\ell}$ is updated by solving
\begin{align}
\underset{\boldsymbol{\Theta}_{\ell}}{\mathsf{\mathrm{max}}}\;\;
&\left\Vert\mathbf{H}_{RI,\ell}\left(\boldsymbol{\Theta}_{\ell}-\mathbf{I}\right)\mathbf{H}_{IT,\ell}\right\Vert^{2}\label{eq:prob1-o}\\
\mathsf{\mathrm{s.t.}}\;\;\;
&\boldsymbol{\Theta}_{\ell}=\text{diag}\left(e^{j\theta_{\ell,1}},e^{j\theta_{\ell,2}},\ldots,e^{j\theta_{\ell,N_I}}\right).\label{eq:prob1-c}
\end{align}
To this end, \eqref{eq:prob1-o}-\eqref{eq:prob1-c} is reformulated into the equivalent problem
\begin{align}
\underset{\mathbf{u},\mathbf{v},\boldsymbol{\Theta}_{\ell}}{\mathsf{\mathrm{max}}}\;\;
&\left\vert\mathbf{u}\mathbf{H}_{RI,\ell}\left(\boldsymbol{\Theta}_{\ell}-\mathbf{I}\right)\mathbf{H}_{IT,\ell}\mathbf{v}\right\vert^{2}\label{eq:prob2-o}\\
\mathsf{\mathrm{s.t.}}\;\;\;
&\eqref{eq:prob1-c},\;\left\Vert\mathbf{u}\right\Vert=1,\;\left\Vert\mathbf{v}\right\Vert=1,\label{eq:prob2-c}
\end{align}
where the auxiliary variables $\mathbf{u}\in\mathbb{C}^{1\times N_I}$ and $\mathbf{v}\in\mathbb{C}^{N_I\times 1}$ have been added.
To solve \eqref{eq:prob2-o}-\eqref{eq:prob2-c}, we initialize $\mathbf{u}$ and $\mathbf{v}$ to feasible values and alternate between the following two steps until convergence of the objective \eqref{eq:prob2-o}:
\textit{i)} with $\mathbf{u}$ and $\mathbf{v}$ fixed, $\boldsymbol{\Theta}_{\ell}$ is updated by solving
\begin{equation}
\underset{\boldsymbol{\Theta}_{\ell}}{\mathsf{\mathrm{max}}}\;
\left\vert g_{RT,\ell}+\mathbf{g}_{RI,\ell}\boldsymbol{\Theta}_{\ell}\mathbf{g}_{IT,\ell}\right\vert^{2}\;
\mathsf{\mathrm{s.t.}}\;
\eqref{eq:prob1-c},\label{eq:prob3}
\end{equation}
where $g_{RT,\ell}=-\mathbf{u}\mathbf{H}_{RI,\ell}\mathbf{H}_{IT,\ell}\mathbf{v}$, $\mathbf{g}_{RI,\ell}=\mathbf{u}\mathbf{H}_{RI,\ell}$, and $\mathbf{g}_{IT,\ell}=\mathbf{H}_{IT,\ell}\mathbf{v}$, whose global optimal solution is $\theta_{\ell,n_I}=\arg(g_{RT,\ell})-\arg([\mathbf{g}_{RI,\ell}]_{n_I})-\arg([\mathbf{g}_{IT,\ell}]_{n_I})$, for $n_I=1,\ldots,N_I$;
\textit{ii)} with $\boldsymbol{\Theta}_{\ell}$ fixed, $\mathbf{u}$ and $\mathbf{v}$ are updated as the dominant left and right singular vectors of $\mathbf{H}_{RI,\ell}\left(\boldsymbol{\Theta}_{\ell}-\mathbf{I}\right)\mathbf{H}_{IT,\ell}$, respectively, which is global optimal.
Although problem \eqref{eq:prob1-o}-\eqref{eq:prob1-c} is written for D-RIS, it can directly be adapted to BD-RIS by substituting in \eqref{eq:prob1-c} the appropriate constraint for the BD-RIS scattering matrix.
Global optimal solutions for \eqref{eq:prob3} adapted to fully-/group-connected RISs and tree-/forest-connected RISs have been derived in \cite{ner22} and \cite{ner23-1}, respectively.
As summarized in Alg.~\ref{alg:opt}, our optimization algorithm solves \eqref{eq:prob1-o}-\eqref{eq:prob1-c} for the $\ell$th RIS iterating from  $\ell=1$ to $L$ until convergence of the objective \eqref{eq:prob1-o} is reached.

\begin{algorithm}[t]
\begin{algorithmic}[1]
\REQUIRE $\mathbf{H}_{RI,L}$, $\mathbf{H}_{IT,1}$, $\mathbf{H}_{\ell+1,\ell}$, for $\ell=1,\ldots,L-1$.
\ENSURE $\boldsymbol{\Theta}_1,\ldots,\boldsymbol{\Theta}_L$.
\STATE{Initialize $\boldsymbol{\Theta}_2,\ldots,\boldsymbol{\Theta}_L$.}
\WHILE{no convergence of objective \eqref{eq:prob1-o}}
\FOR{$\ell=1$ \textbf{to} $L$}
\STATE{Update $\mathbf{H}_{RI,\ell}$ by \eqref{eq:hRIl} if $\ell\neq L$.}
\STATE{Update $\mathbf{H}_{IT,\ell}$ by \eqref{eq:hITl} if $\ell\neq 1$.}
\STATE{Update $\boldsymbol{\Theta}_\ell$ by solving \eqref{eq:prob1-o}-\eqref{eq:prob1-c}.}
\ENDFOR
\ENDWHILE
\STATE{Return $\boldsymbol{\Theta}_1,\ldots,\boldsymbol{\Theta}_L$.}
\end{algorithmic}
\caption{Optimization of multi-RIS aided systems with multipath channels}
\label{alg:opt}
\end{algorithm}

To verify the effectiveness of Alg.~\ref{alg:opt}, we derive an upper bound on the achievable channel gain $\Vert\mathbf{H}\Vert^{2}$.
For this purpose, we notice that the physics-compliant channel in \eqref{eq:H} can be rewritten as a sum of $2^L$ terms as
\begin{equation}
\mathbf{H}=\sum_{i=1}^{2^L}\mathbf{H}^{(i)},
\end{equation}
where
\begin{equation}
\mathbf{H}^{(i)}=\left(-1\right)^{L-W^{(i)}}
\mathbf{H}_{RI,L}\boldsymbol{\Theta}_{L}^{b_L^{(i)}}
\prod_{\ell=L-1}^{1}\left(\mathbf{H}_{\ell+1,\ell}\boldsymbol{\Theta}_{\ell}^{b_{\ell}^{(i)}}\right)\mathbf{H}_{IT,1},
\end{equation}
with $\mathbf{b}^{(i)}=[b_1^{(i)},b_2^{(i)},\ldots,b_L^{(i)}]\in\{0,1\}^L$ being the $i$th $L$-bit binary vector and $W^{(i)}$ being the number of 1s in $\mathbf{b}^{(i)}$, and where $\boldsymbol{\Theta}_{\ell}^0=\mathbf{I}$ and $\boldsymbol{\Theta}_{\ell}^1=\boldsymbol{\Theta}_{\ell}$.
Thus, by the triangle inequality, $\Vert\mathbf{H}\Vert^2$ can be upper bounded by
\begin{equation}
\left\Vert\mathbf{H}\right\Vert^2\leq\left(\sum_{i=1}^{2^L}\left\Vert\mathbf{H}^{(i)}\right\Vert\right)^2,\label{eq:UB1}
\end{equation}
where $\Vert\mathbf{H}^{(i)}\Vert$ are upper bounded depending on the value of $W^{(i)}$, as explained in the following.
First, if $W^{(i)}=0$, i.e., there is no scattering matrix $\boldsymbol{\Theta}_{\ell}$ in $\mathbf{H}^{(i)}$, we have
\begin{equation}
\left\Vert\mathbf{H}^{(i)}\right\Vert=\left\Vert\mathbf{H}_{RI,L}\mathbf{H}_{L,L-1}\cdots\mathbf{H}_{2,1}\mathbf{H}_{IT,1}\right\Vert.\label{eq:UB2}
\end{equation}
Second, if $W^{(i)}=1$, i.e., there is only one scattering matrix in $\mathbf{H}^{(i)}$, denoted as $\boldsymbol{\Theta}_{u}$, it holds
\begin{multline}
\left\Vert\mathbf{H}^{(i)}\right\Vert\leq\left\Vert\mathbf{H}_{RI,L}\mathbf{H}_{L,L-1}\cdots\mathbf{H}_{u+1,u}\right\Vert\\
\times\left\Vert\mathbf{H}_{u,u-1}\cdots\mathbf{H}_{2,1}\mathbf{H}_{IT,1}\right\Vert,\label{eq:UB3}
\end{multline}
following the submultiplicative property of the spectral norm, and that $\Vert\boldsymbol{\Theta}_u\Vert=1$ for any unitary $\boldsymbol{\Theta}_u$.
Third, if $W^{(i)}\geq2$, i.e., there are multiple scattering matrices $\boldsymbol{\Theta}_{\ell}$ in $\mathbf{H}^{(i)}$, with $\ell\in\{u_1<u_2<\ldots<u_{W^{(i)}}\}$, we have
\begin{multline}
\left\Vert\mathbf{H}^{(i)}\right\Vert\leq\left\Vert\mathbf{H}_{RI,L}\mathbf{H}_{L,L-1}\cdots\mathbf{H}_{u_{W^{(i)}}+1,u_{W^{(i)}}}\right\Vert\\
\times\prod_{w=W^{(i)}-1}^{1}\left(\left\Vert\mathbf{H}_{u_{w+1},u_{w+1}-1}\cdots\mathbf{H}_{u_w+1,u_w}\right\Vert\right)\\
\times\left\Vert\mathbf{H}_{u_1,u_1-1}\cdots\mathbf{H}_{2,1}\mathbf{H}_{IT,1}\right\Vert,\label{eq:UB4}
\end{multline}
following the submultiplicative property of the spectral norm, and that $\Vert\boldsymbol{\Theta}_\ell\Vert=1$, $\forall\ell$.

To better visualize the upper bound given in \eqref{eq:UB1}-\eqref{eq:UB4}, consider a multi-RIS aided system with $L=2$ RISs.
In this case, the physics-compliant channel is
\begin{equation}
\mathbf{H}
=\mathbf{H}_{RI,2}\left(\boldsymbol{\Theta}_2-\mathbf{I}\right)\mathbf{H}_{2,1}\left(\boldsymbol{\Theta}_1-\mathbf{I}\right)\mathbf{H}_{IT,1},
\end{equation}
which can be rewritten as a sum of $2^L=4$ terms as
\begin{multline}
\mathbf{H}
=\mathbf{H}_{RI,2}\boldsymbol{\Theta}_2\mathbf{H}_{2,1}\boldsymbol{\Theta}_1\mathbf{H}_{IT,1}
-\mathbf{H}_{RI,2}\boldsymbol{\Theta}_2\mathbf{H}_{2,1}\mathbf{H}_{IT,1}\\
-\mathbf{H}_{RI,2}\mathbf{H}_{2,1}\boldsymbol{\Theta}_1\mathbf{H}_{IT,1}
+\mathbf{H}_{RI,2}\mathbf{H}_{2,1}\mathbf{H}_{IT,1}.\label{eq:H2}
\end{multline}
Thus, by using the triangle inequality and individually upper bounding each term, we obtain
\begin{multline}
\left\Vert\mathbf{H}\right\Vert^2
\leq\left(\left\Vert\mathbf{H}_{RI,2}\right\Vert\left\Vert\mathbf{H}_{2,1}\right\Vert\left\Vert\mathbf{H}_{IT,1}\right\Vert
+\left\Vert\mathbf{H}_{RI,2}\right\Vert\left\Vert\mathbf{H}_{2,1}\mathbf{H}_{IT,1}\right\Vert\right.\\
+\left.\left\Vert\mathbf{H}_{RI,2}\mathbf{H}_{2,1}\right\Vert\left\Vert\mathbf{H}_{IT,1}\right\Vert
+\left\Vert\mathbf{H}_{RI,2}\mathbf{H}_{2,1}\mathbf{H}_{IT,1}\right\Vert\right)^2,
\end{multline}
which is the upper bound in \eqref{eq:UB1}-\eqref{eq:UB4} for the case $L=2$.

\subsection{Widely Used Channel Model}

For the widely used channel model, the optimization algorithm in Alg.~\ref{alg:opt} can still be applied considering two modifications.
First, the equivalent channels $\mathbf{H}_{RI,\ell}$ and $\mathbf{H}_{IT,\ell}$ in \eqref{eq:hRIl} and \eqref{eq:hITl} are now defined as
$\mathbf{H}_{RI,\ell}=\mathbf{H}_{RI,L}\prod_{k=L}^{\ell+1}(\boldsymbol{\Theta}_{k}\mathbf{H}_{k,k-1})$ and
$\mathbf{H}_{IT,\ell}=\prod_{k=\ell-1}^{1}(\mathbf{H}_{k+1,k}\boldsymbol{\Theta}_{k})\mathbf{H}_{IT,1}$.
Second, the objective of problem \eqref{eq:prob1-o}-\eqref{eq:prob1-c} is now given by $\Vert\mathbf{H}_{RI,\ell}\boldsymbol{\Theta}_{\ell}\mathbf{H}_{IT,\ell}\Vert^{2}$.
This new problem can still be solved by introducing the auxiliary variables $\mathbf{u}$ and $\mathbf{v}$ and alternating between the following two steps:
\textit{i)} with $\mathbf{u}$ and $\mathbf{v}$ fixed, $\boldsymbol{\Theta}_{\ell}$ is updated by solving
\begin{equation}
\underset{\boldsymbol{\Theta}_{\ell}}{\mathsf{\mathrm{max}}}\;
\left\vert\mathbf{g}_{RI,\ell}\boldsymbol{\Theta}_{\ell}\mathbf{g}_{IT,\ell}\right\vert^{2}\;
\mathsf{\mathrm{s.t.}}\;
\eqref{eq:prob1-c},
\end{equation}
where $\mathbf{g}_{RI,\ell}=\mathbf{u}\mathbf{H}_{RI,\ell}$ and $\mathbf{g}_{IT,\ell}=\mathbf{H}_{IT,\ell}\mathbf{v}$, whose global optimal solution is $\theta_{\ell,n_I}=-\arg([\mathbf{g}_{RI,\ell}]_{n_I})-\arg([\mathbf{g}_{IT,\ell}]_{n_I})$, for $n_I=1,\ldots,N_I$, for D-RIS, and given in \cite{ner22,ner23-1} for BD-RIS;
\textit{ii)} with $\boldsymbol{\Theta}_{\ell}$ fixed, $\mathbf{u}$ and $\mathbf{v}$ are updated as the dominant left and right singular vectors of $\mathbf{H}_{RI,\ell}\boldsymbol{\Theta}_{\ell}\mathbf{H}_{IT,\ell}$, respectively.

Besides, the channel gain $\Vert\mathbf{H}^\prime\Vert^2$ is upper bounded by
\begin{equation}
\left\Vert\mathbf{H}^\prime\right\Vert^2
\leq\left\Vert\mathbf{H}_{RI,L}\right\Vert^2
\prod_{\ell=L-1}^{1}\left(\left\Vert\mathbf{H}_{\ell+1,\ell}\right\Vert^2\right)\left\Vert\mathbf{H}_{IT,1}\right\Vert^2,\label{eq:UBprime}
\end{equation}
applying the submultiplicativity of the spectral norm, and that $\Vert\boldsymbol{\Theta}_\ell\Vert=1$, $\forall\ell$.
For example, in the case $L=2$, the gain of the widely used channel is upper bounded by $\Vert\mathbf{H}^\prime\Vert^2\leq\Vert\mathbf{H}_{RI,2}\Vert^2\Vert\mathbf{H}_{2,1}\Vert^2\Vert\mathbf{H}_{IT,1}\Vert^2$.

\section{Numerical Results}
\label{sec:results}

In this section, we numerically quantify the difference between the channel models in \eqref{eq:H} and \eqref{eq:Hprime}, and validate the theoretical insights obtained in Sections~\ref{sec:gain-los} and \ref{sec:gain-multipath} for \gls{los} and multipath channels, respectively.
To this end, we introduce the relative difference between the average channel gain of the physics-compliant model $\text{E}[\Vert\mathbf{H}\Vert^2]$ and the widely used model $\text{E}[\Vert\mathbf{H}^\prime\Vert^2]$ as
\begin{equation}
\eta=\frac{\text{E}\left[\left\Vert\mathbf{H}\right\Vert^2\right]-\text{E}\left[\left\Vert\mathbf{H}^\prime\right\Vert^2\right]}{\text{E}\left[\left\Vert\mathbf{H}^\prime\right\Vert^2\right]}.\label{eq:eta1}
\end{equation}
Furthermore, as the physics-compliant and the widely used models are different, optimizing the RISs based on the widely used model results in performance degradation when testing on the physics-compliant model.
To assess such a performance degradation, we consider the normalized gain
\begin{equation}
\rho=\frac{\text{E}\left[\left\Vert\mathbf{H}^{\text{Sub}}\right\Vert^2\right]}{\text{E}\left[\left\Vert\mathbf{H}\right\Vert^2\right]},\label{eq:rho1}
\end{equation}
where $\text{E}[\Vert\mathbf{H}^{\text{Sub}}\Vert^2]$ is the average channel gain of the physics-compliant model obtained by optimizing the RISs based on the widely used model, giving a suboptimal solution.
In the following, we separately analyze multi-RIS aided systems under \gls{los} and multipath channels, by setting $N_R=N_T=2$.

\subsection{Line-of-Sight Channels}

In Fig.~\ref{fig:LoS}, we report the theoretical scaling laws of the physics-compliant and widely used models given by \eqref{eq:EG} and \eqref{eq:Gprime}, respectively, and the corresponding channel gains obtained through Monte Carlo simulations using the proposed globally optimal optimization strategy, with unit path gain, i.e., $\Lambda=1$.
In our simulations, we consider $\alpha_{RI,L,n_R}$, $\beta_{RI,L,n_I}$, $\alpha_{IT,1,n_I}$, $\beta_{IT,1,n_T}$, $\alpha_{\ell+1,\ell,n_I}$, and $\beta_{\ell+1,\ell,n_I}$ independent and uniformly distributed in $[0,2\pi]$.
We observe that the theoretical scaling laws are accurate as they coincide with the numerical simulations.
Furthermore, the physics-compliant model results in a higher channel gain than the widely used model because of the structural scattering terms, which can be exploited to increase the channel gain.

By substituting the scaling laws \eqref{eq:EG} and \eqref{eq:Gprime} into \eqref{eq:eta1}, we obtain the relative difference between $\text{E}[\Vert\mathbf{H}\Vert^2]$ and $\text{E}[\Vert\mathbf{H}^\prime\Vert^2]$ under \gls{los} channels as
\begin{equation}
\eta=\frac{\left(N_I+\sqrt{\pi N_I}+1\right)^L-N_I^{L}}{N_I^L},\label{eq:eta2}
\end{equation}
in closed form.
In Fig.~\ref{fig:LoS-diff}, we report the relative difference $\eta$ given theoretically in \eqref{eq:eta2} and obtained as a result of Monte Carlo simulations.
We observe that the theoretical insights are confirmed by the numerical results.
In addition, we make the following two observations.
\textit{First}, the relative difference decreases with $N_I$ since in \eqref{eq:EG} the structural scattering term scales with $N_I$ while the RIS-aided term scales with $N_I^2$.
However, the relative difference is non-negligible for a practical number of RIS elements.
Specifically, considering $L=4$ RISs, the relative difference is higher than 400\% when $N_I=16$ and higher than 80\% when $N_I=128$.
\textit{Second}, the relative difference increases with $L$ since each RIS contributes its structural scattering, not included in the widely used model.

\begin{figure}[t]
\centering
\subfigure[$L=2$]{
\includegraphics[width=0.23\textwidth]{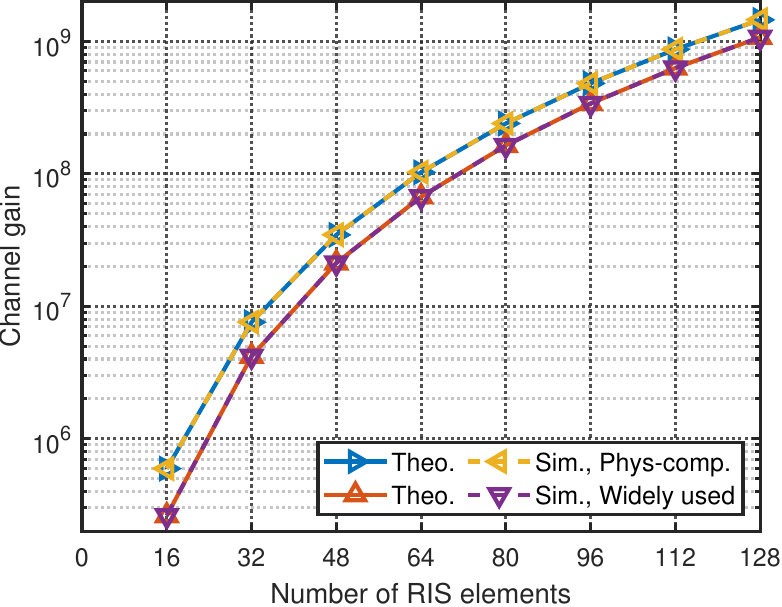}}
\subfigure[$L=4$]{
\includegraphics[width=0.23\textwidth]{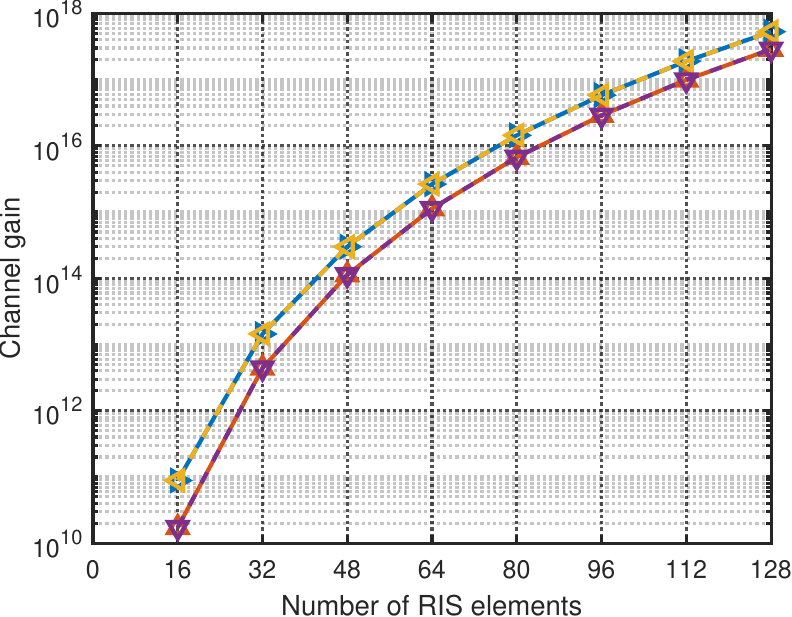}}
\caption{Average channel gain with unit path gain, i.e., $\Lambda=1$, of the physics-compliant model $\mathbf{H}$ and the widely used model $\mathbf{H}^\prime$, under LoS channels.}
\label{fig:LoS}
\end{figure}

We now quantify the performance degradation due to optimizing the RISs based on the widely used model.
By plugging the solution \eqref{eq:sol-wu1}-\eqref{eq:sol-wu3} (optimal for the widely used model) into the physics-compliant model \eqref{eq:Hlos}, we obtain the average suboptimal channel gain as
\begin{align}
\text{E}\left[\left\Vert\mathbf{H}^{\text{Sub}}\right\Vert^2\right]
&=\Lambda^2\text{E}\left[\left\vert c+N_I\right\vert^2\right]^LN_RN_T\\
&=\Lambda^2\text{E}\left[\left\vert c\right\vert^2+N_I^2+2N_I\Re\{c\}\right]^LN_RN_T,
\end{align}
where $c\sim\mathcal{CN}(0,N_I)$.
Consequently, noticing that $\text{E}[\vert c\vert^2]=N_I$ and $\text{E}[\Re\{c\}]=0$, we obtain
\begin{equation}
\text{E}\left[\left\Vert\mathbf{H}^{\text{Sub}}\right\Vert^2\right]
=\Lambda^2\left(N_I^2+N_I\right)^LN_RN_T,\label{eq:EGsub}
\end{equation}
and, by substituting \eqref{eq:EG} and \eqref{eq:EGsub} into \eqref{eq:rho1}, we obtain the normalized channel gain under \gls{los} channels as
\begin{equation}
\rho=\left(\frac{N_I+1}{N_I+\sqrt{\pi N_I}+1}\right)^L.\label{eq:rho2}
\end{equation}
In Fig.~\ref{fig:LoS-gain}, we report the theoretical normalized gain $\rho$ given in \eqref{eq:rho2} together with the simulated one, and we observe that they are equivalent.
In addition, we notice that the normalized gain decreases with $L$ since the discrepancy between the physics-compliant and the widely used model increases with $L$.
With $L=4$ RISs, the normalized gain is 25\% when $N_I=16$ and 56\% when $N_I=128$, showing an important degradation when optimizing the RISs based on the widely used model.

\begin{figure}[t]
\centering
\includegraphics[width=0.38\textwidth]{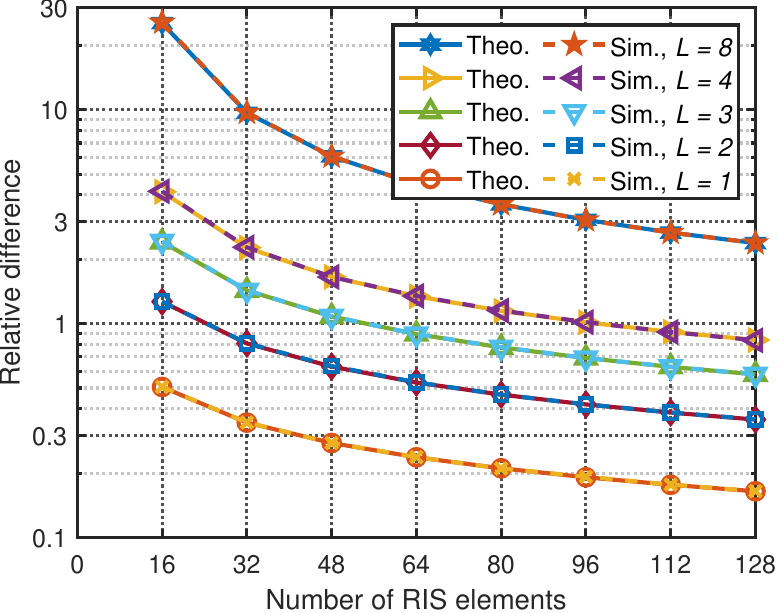}
\caption{Relative difference $\eta$ between the average channel gain of the physics-compliant model $\mathbf{H}$ and the widely used model $\mathbf{H}^\prime$, under LoS channels.}
\label{fig:LoS-diff}
\end{figure}
\begin{figure}[t]
\centering
\includegraphics[width=0.38\textwidth]{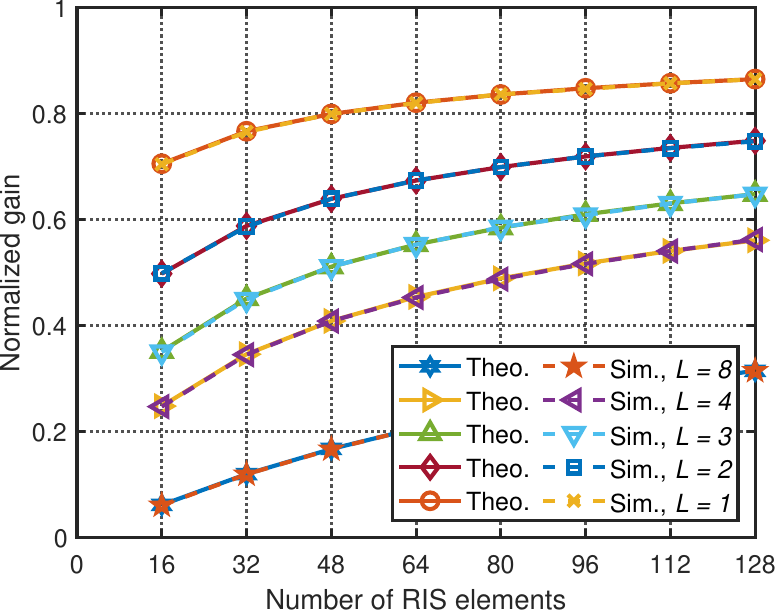}
\caption{Normalized gain $\rho$ of the physics-compliant model $\mathbf{H}$ obtained by optimizing the RISs based on the widely used model $\mathbf{H}^\prime$, under LoS channels.}
\label{fig:LoS-gain}
\end{figure}

\subsection{Multipath Channels}

In Fig.~\ref{fig:multipath}, we report the channel gains of the physics-compliant and widely used channel models, under Rayleigh channels with unit path gain.
For both the physics-compliant and widely used channel models we report: \textit{i)} the channel gain upper bounds given by \eqref{eq:UB1}-\eqref{eq:UB4} and \eqref{eq:UBprime}, respectively, which serve as a benchmark for the proposed optimization algorithm, \textit{ii)} the channel gains obtained by optimizing BD-RISs with Alg.~\ref{alg:opt}, where the tree-connected BD-RIS architecture is considered \cite{ner23-1}, and \textit{iii)} the channel gains obtained by optimizing D-RISs with Alg.~\ref{alg:opt}.
We make the following three observations.
\textit{First}, the physics-compliant model results in a higher channel gain than the widely used model, as observed for \gls{los} channels in Fig.~\ref{fig:LoS}.
\textit{Second}, BD-RIS allows to reach a higher channel gain than D-RIS, given its additional flexibility \cite{she20,ner23-1}.
\textit{Third}, the channel gain upper bound for the widely used model in \eqref{eq:UBprime} is tight and is achieved by optimizing BD-RISs, while the upper bound on the physics-compliant model is not tight due to the sub-optimality of the optimization algorithm Alg.~\ref{alg:opt}.

\begin{figure}[t]
\centering
\subfigure[$L=2$]{
\includegraphics[width=0.23\textwidth]{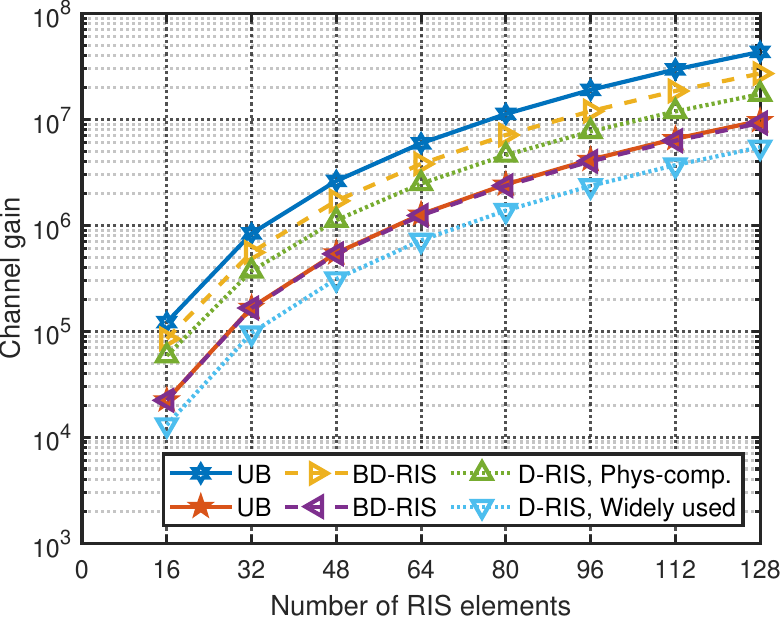}}
\subfigure[$L=4$]{
\includegraphics[width=0.23\textwidth]{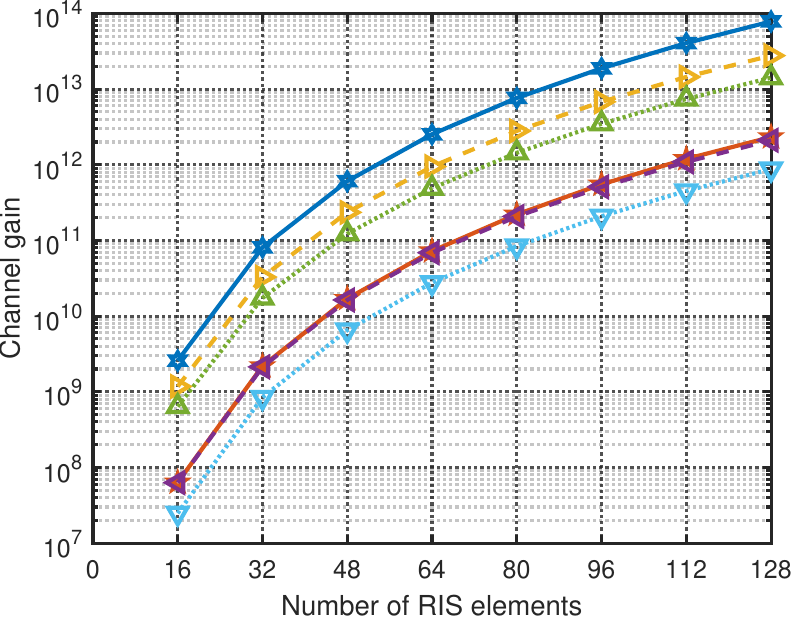}}
\caption{Average channel gain with unit path gain, i.e., $\Lambda=1$, of the physics-compliant model $\mathbf{H}$ and the widely used model $\mathbf{H}^\prime$, under Rayleigh channels.}
\label{fig:multipath}
\end{figure}

In Fig.~\ref{fig:multipath-diff}, we report the simulated relative difference $\eta$ under Rayleigh channels, in the presence of BD-RISs and D-RISs.
Similarly to what was observed in Fig.~\ref{fig:LoS-diff} for \gls{los} channels, we notice that the relative difference decreases with $N_I$ and increases with $L$, in both BD-RIS and D-RIS aided systems.
Interestingly, this relative difference is non-negligible, confirming the need for accurate models of multi-RIS aided wireless channels.
For example, considering $L=4$ RISs with $N_I=128$ elements each, the relative difference is higher than 1000\% in both BD-RIS and D-RIS aided systems.

In Fig.~\ref{fig:multipath-gain}, we report the simulated normalized gain $\rho$ under Rayleigh channels.
As noticed for \gls{los} channels in Fig.~\ref{fig:LoS-gain}, we observe that the normalized gain decreases with $L$.
In the case of a system aided by $L=4$ D-RISs with $N_I=128$ elements each, the normalized gain is 0.07, indicating that only 7\% of the maximum channel gain can be achieved by optimizing D-RISs based on the widely used channel model.

In Fig.~\ref{fig:multipath-Rician-diff}, we report the relative difference between the channel gain with the physics-compliant and the widely used models, under Rician channels for different values of Rician factor $K$, to verify how the relative difference varies with the level of multipath richness.
Interestingly, the relative difference increases as the Rician factor $K$ decreases, for both $L=2$ and $L=4$, and for both BD-RIS and D-RIS aided systems.
This means that the impact of the specular reflection caused by the structural scattering on the end-to-end channel is lower under rank-1 \gls{los} channels and higher in the presence of multipath channels.
In the following, we provide a physical and a mathematical explanation for this observation.

From the physical perspective, with rank-1 \gls{los} channels, the RISs are connected by wireless channel matrices having only one useful eigenmode, i.e., corresponding to the only non-zero channel eigenvalue.
Thus, the end-to-end channel is impacted only by the component of the specular reflection at each RIS aligned with such an eigenmode.
Conversely, under multipath propagation, the RISs are connected by wireless channel matrices having multiple useful eigenmodes, as their rank is greater than 1.
In this case, all the components of the specular reflection at each RIS aligned with these eigenmodes are captured in the end-to-end channel, increasing the impact of the specular reflections.
Mathematically, we justify this observation by rigorously analyzing the cascaded structural scattering term in a multi-RIS aided \gls{siso} channel with $L=2$ RISs.
To this end, consider the multi-RIS aided channel with $L=2$ RISs given in \eqref{eq:H2}, where the term accounting for the cascade of the structural scattering of both RISs is $\mathbf{h}_{RI,2}\mathbf{H}_{2,1}\mathbf{h}_{IT,1}$ in the \gls{siso} case.
Thus, the average structural scattering strength, normalized by the widely used channel gain upper bound, writes as
\begin{align}
s
=&\frac{\text{E}\left[\left\vert\mathbf{h}_{RI,2}\mathbf{H}_{2,1}\mathbf{h}_{IT,1}\right\vert^2\right]}
{\text{E}\left[\left\Vert\mathbf{h}_{RI,2}\right\Vert^2\left\Vert\mathbf{H}_{2,1}\right\Vert^2\left\Vert\mathbf{h}_{IT,1}\right\Vert^2\right]}.
\label{eq:s1}
\end{align}
Assuming that $\text{E}[\Re\{([\mathbf{h}_{RI,2}]_{i}[\mathbf{H}_{2,1}]_{i,j}[\mathbf{h}_{IT,1}]_{j})$ $\times$ $
([\mathbf{h}_{RI,2}]_{i^\prime}[\mathbf{H}_{2,1}]_{i^\prime,j^\prime}[\mathbf{h}_{IT,1}]_{j^\prime})^*\}]=0$, for $(i,j)\neq(i^\prime,j^\prime)$, \eqref{eq:s1} simplifies as
\begin{align}
s
&=\sum_{i,j=1}^{N_I}\frac{\text{E}\left[\left\vert\left[\mathbf{h}_{RI,2}\right]_{i}\left[\mathbf{H}_{2,1}\right]_{i,j}\left[\mathbf{h}_{IT,1}\right]_{j}\right\vert^2\right]}
{\text{E}\left[\left\Vert\mathbf{h}_{RI,2}\right\Vert^2\left\Vert\mathbf{H}_{2,1}\right\Vert^2\left\Vert\mathbf{h}_{IT,1}\right\Vert^2\right]}.
\label{eq:s2}
\end{align}
In addition, assuming that $\mathbf{h}_{RI,2}$, $\mathbf{H}_{2,1}$, and $\mathbf{h}_{IT,1}$ are independent, and that the entries of $\mathbf{h}_{RI,2}$ are identically distributed as well as the entries of $\mathbf{h}_{IT,1}$, \eqref{eq:s2} can be rewritten as
\begin{equation}
s=\frac{1}{N_I^2}\frac{\text{E}\left[\left\Vert\mathbf{H}_{2,1}\right\Vert_F^2\right]}{\text{E}\left[\left\Vert\mathbf{H}_{2,1}\right\Vert^2\right]}.\label{eq:s3}
\end{equation}
Denoting as $\sigma_{n_I}$ the $n_I$th singular value of $\mathbf{H}_{2,1}$, and introducing $\bar{\lambda}_{n_I}=\text{E}[\sigma_{n_I}^2]$ for ${n_I}=1,\ldots,N_I$, we can write $\text{E}[\Vert\mathbf{H}_{2,1}\Vert_F^2]=\sum_{{n_I}=1}^{N_I}\bar{\lambda}_{n_I}$ and $\text{E}[\Vert\mathbf{H}_{2,1}\Vert^2]=\bar{\lambda}_1$, yielding
\begin{equation}
s=\frac{1}{N_I^2}\left(1+\sum_{{n_I}=2}^{N_I}\frac{\bar{\lambda}_{n_I}}{\bar{\lambda}_1}\right).\label{eq:s4}
\end{equation}
Remarkably, \eqref{eq:s4} is minimum under \gls{los} channels, being $s=1/N_I^2$ since $\bar{\lambda}_{n_I}/\bar{\lambda}_1=0$, for ${n_I}=2,\ldots,N_I$, in this case.
Besides, the terms $\bar{\lambda}_{n_I}/\bar{\lambda}_1$ increase with the level of multipath richness of the channel $\mathbf{H}_{2,1}$, leading to an increase in the structural scattering strength $s$ as the Rician factor $K$ decreases.
For this reason, we observe that the relative difference between physics-compliant and widely used models increases as the Rician factor $K$ decreases in Fig.~\ref{fig:multipath-Rician-diff}.

\begin{figure}[t]
\centering
\includegraphics[width=0.38\textwidth]{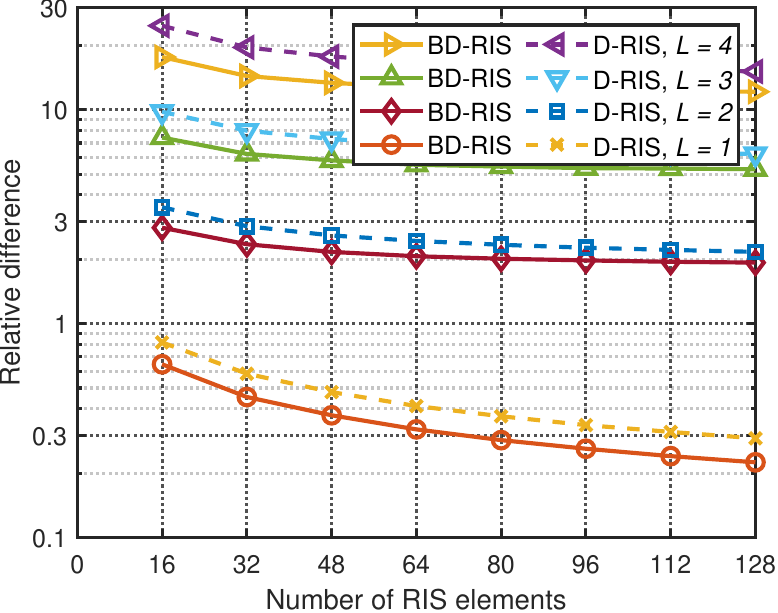}
\caption{Relative difference $\eta$ between the average channel gain of the physics-compliant model $\mathbf{H}$ and the widely used model $\mathbf{H}^\prime$, under Rayleigh channels.}
\label{fig:multipath-diff}
\end{figure}
\begin{figure}[t]
\centering
\includegraphics[width=0.38\textwidth]{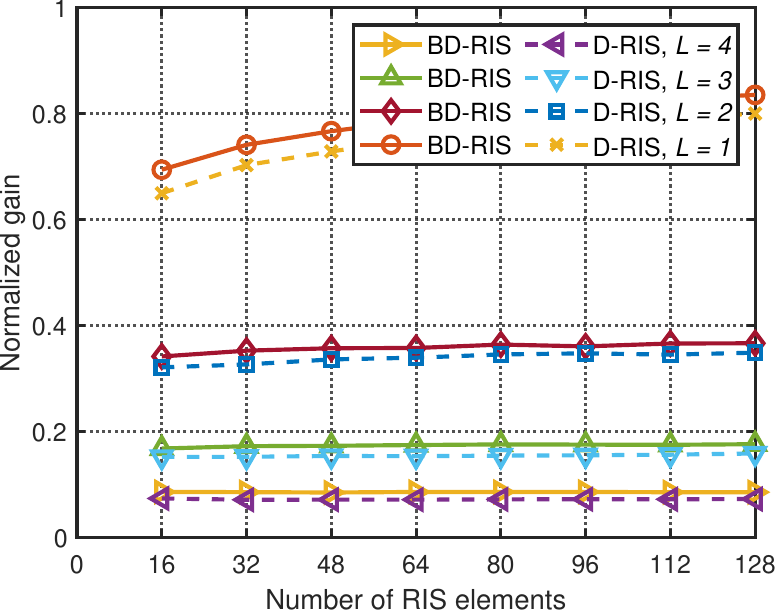}
\caption{Normalized gain $\rho$ of the physics-compliant model $\mathbf{H}$ obtained by optimizing the RISs based on the widely used model $\mathbf{H}^\prime$, under Rayleigh channels.}
\label{fig:multipath-gain}
\end{figure}

\begin{figure*}[t]
\centering
\subfigure[$L=2$]{
\includegraphics[width=0.38\textwidth]{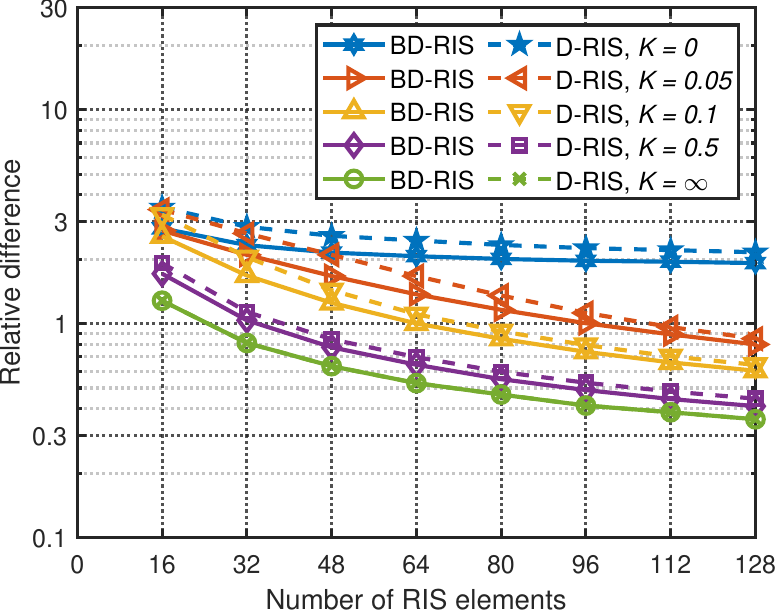}}
\subfigure[$L=4$]{
\includegraphics[width=0.38\textwidth]{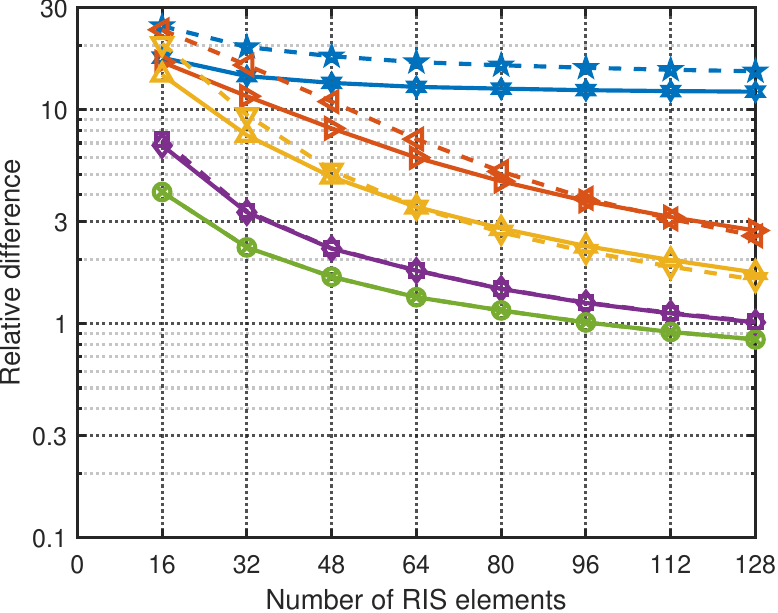}}
\caption{Relative difference $\eta$ between the average channel gain of the physics-compliant model $\mathbf{H}$ and the widely used model $\mathbf{H}^\prime$, under Rician channels.}
\label{fig:multipath-Rician-diff}
\end{figure*}
\begin{figure*}[t]
\centering
\subfigure[$L=2$]{
\includegraphics[width=0.38\textwidth]{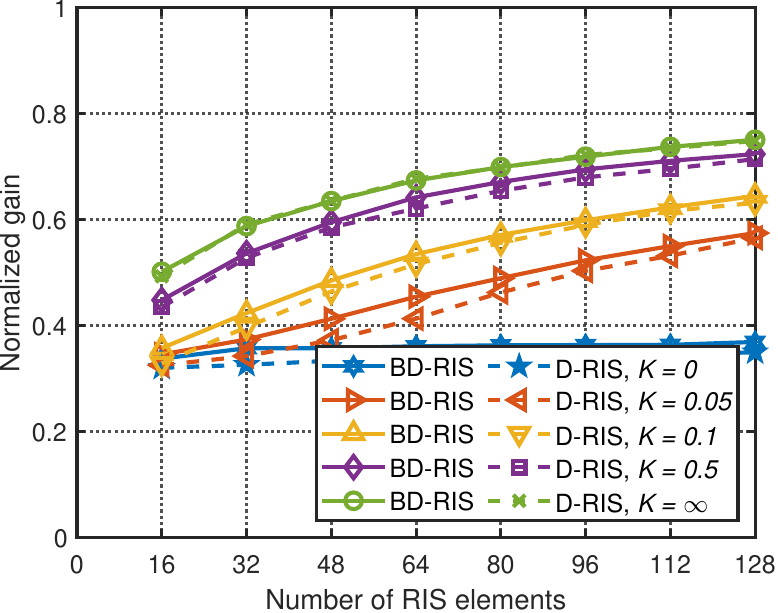}}
\subfigure[$L=4$]{
\includegraphics[width=0.38\textwidth]{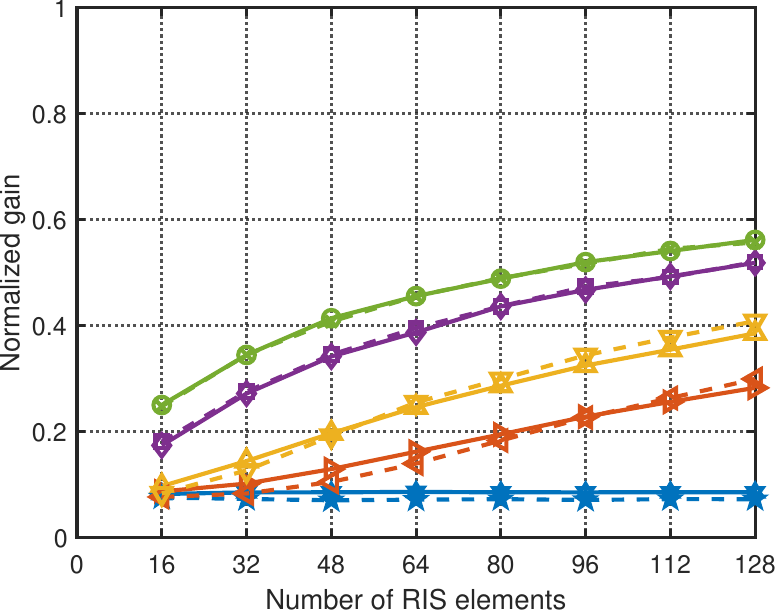}}
\caption{Normalized gain $\rho$ of the physics-compliant model $\mathbf{H}$ obtained by optimizing the RISs based on the widely used model $\mathbf{H}^\prime$, under Rician channels.}
\label{fig:multipath-Rician-gain}
\end{figure*}

Finally, in Fig.~\ref{fig:multipath-Rician-gain}, we show how the normalized gain varies with the multipath richness of the channels by reporting the normalized gain obtained under Rician channels.
The normalized gain of the physics-compliant model obtained by optimizing the RISs based on the widely used model decreases with the Rician factor $K$, consistently with the fact that the relative difference between the two models increases as the Rician factor $K$ decreases, as observed in Fig.~\ref{fig:multipath-Rician-diff}.


\section{Conclusion}
\label{sec:conclusion}

We present a physics-compliant channel model for multi-RIS aided systems that accurately characterizes multi-hop communication links enabled by multiple cooperative RISs.
The obtained model, derived through rigorous multiport network theory, aligns with the model widely used in related literature, while accounting for the structural scattering of the RISs, commonly neglected in the literature.
We show that reflective RISs exhibit structural scattering impacting the channel model, whereas transmissive RISs do not.

Since the physics-compliant channel model and the widely used one are different in the presence of reflective RISs, we compare them considering \gls{los} as well as multipath channels.
Under \gls{los} channels, we characterize the scaling laws of the channel gains for the two models.
Considering multipath channels, we maximize the channel gains for the two models by optimizing D-RISs and BD-RISs, and propose closed-form upper bounds on their channel gains.
Theoretical derivations, corroborated by numerical simulations, show that the physics-compliant channel gain substantially differs from the widely used one, and that their relative difference increases with the number of RISs in the system.
Furthermore, this discrepancy is more pronounced with multipath channels than with \gls{los} channels.
In a system aided by four 128-element RISs with \gls{los} channels, if the RISs are optimized using the widely used model and their solutions are applied to the physics-compliant channel model, it can be achieved only 56\% of the maximum achievable channel gain.
This value becomes only 7\% in the presence of multipath channels.

While this study provides a theoretical analysis and guidance for the deployment of multi-RIS systems, several avenues for further research can be identified.
Given the important gap between the physics-compliant channel model and the widely used model, all system design choices should be carried out by considering the physics-compliant model.
Potentially, all the problems of multi-RIS systems addressed in previous literature using the widely used channel model should be reconsidered using the physics-compliant model.
Examples of such problems include, but are not limited to: RIS reconfiguration, RIS placement, channel estimation, and beam routing, i.e., selection of the RISs involved in the signal reflection among all those available in the environment.
In addition, further research is needed to practically implement multi-RIS systems and experimentally validate the performance of RIS.

\bibliographystyle{IEEEtran}
\bibliography{IEEEabrv,main}

\end{document}